\documentclass[a4paper,USenglish,cleveref, autoref, thm-restate]{lipics-v2021}

\hideLIPIcs  

\usepackage{latexsym,amsthm,amsmath,color,amssymb,url,hyperref}
\usepackage{tikz}
\usepackage[ruled,linesnumbered]{algorithm2e}
\usepackage[utf8]{inputenc}
\usepackage{color}
\usepackage{cleveref}

\usetikzlibrary{arrows,decorations.pathreplacing,shapes}

\input{colordvi}

\newcommand{\mynewtheorem}[2]{
  \newaliascnt{#1}{dummy}
  \newtheorem{#1}[#1]{#2}
  \aliascntresetthe{#1}
  \expandafter\def\csname #1autorefname\endcsname{#2}
}

\newcommand{\etal}{\emph{et al.}}

\newcommand{\V}{\mathcal{V}}

\newcommand{\F}{\mathcal{F}}
\newcommand{\GLT}{\textsc{GLT}}
\newcommand{\CSC}{\textsc{CSC}}
\newcommand{\Gr}{\mathsf{Gr}}
\newcommand{\ST}{\mathbf{ST}}
\newcommand{\PC}{\mathbf{PC}}
\newcommand{\C}{\mathbf{C}}
\newcommand{\CD}{\mathbf{C}}
\newcommand{\LL}{\mathbf{L}}
\newcommand{\node}{\mathbf{n}}
\newcommand{\cl}{\mathsf{cl}}
\newcommand{\queue}{\mathsf{Active}}
\newcommand{\visited}{\mathsf{Visited}}

\title{Circle graphs  can be recognized in linear time} 

\author{Christophe Paul}{CNRS, Université Montpellier, France \and \url{http://www.lirmm.fr/~paul} }{christophe.paul@lirmm.fr}{https://orcid.org/0000-0001-6519-975X}{Research supported by ANR-DFG project GODASse  ANR-24-CE48-4377}

\author{Ignaz Rutter}{University of Passau, Germany.}{rutter@fim.uni-passau.de}{https://orcid.org/0000-0002-3794-4406}{}

\authorrunning{C. Paul and I. Rutter} 

\Copyright{Christophe Paul and Ignaz Rutter} 



\category{} 

\acknowledgements{}

\nolinenumbers 

\EventEditors{Meena Mahajan, Florin Manea, Annabelle McIver, and Nguy\~{\^{e}}n Kim Th\'{\u{a}}ng}
\EventNoEds{4}
\EventLongTitle{43rd International Symposium on Theoretical Aspects of Computer Science (STACS 2026)}
\EventShortTitle{STACS 2026}
\EventAcronym{STACS}
\EventYear{2026}
\EventDate{March 9--September 13, 2026}
\EventLocation{Grenoble, France}
\EventLogo{}
\SeriesVolume{364}
\ArticleNo{35}

\begin{document}

\maketitle

\begin{abstract}
To date, the best circle graph recognition algorithm, due to Gioan \etal~\cite{GioanPTC14Circle} runs in almost linear time as it relies on a split decomposition algorithm~\cite{GioanPTC14Split} that uses the union-find data-structure~\cite{GallerF64AnImproved,Tarjan75Efficiency}. We show that in the case of circle graphs, the PC-tree data-structure~\cite{ShihH99Anew} allows one to avoid the union-find data-structure to compute the split decomposition in linear time. As a consequence, we obtain the first linear-time recognition algorithm for circle graphs.

To date, the best circle graph recognition algorithm [Gioan, Paul, Tedder, and Corneil. Practical and efficient circle graph recognition. Algorithmica, 69:759–788, 2014] runs in almost linear time as it relies on a split decomposition algorithm [Gioan, Paul, Tedder, and Corneil. Practical and efficient split decomposition via graph-labelled trees. Algorithmica, 69:789–843, 2014] that uses the union-find data-structure. We show that in the case of circle graphs, the PC-tree data-structure [Shih and Hsu. A new planarity test. Theoretical Computer Science, 223:179–191, 1999] allows one to avoid the union-find data-structure to compute the split decomposition in linear time. As a consequence, we obtain the first linear-time recognition algorithm for circle graphs.
\end{abstract}

\section{Introduction}

A \emph{circle graph} is the intersection graph of a set of chords inscribed in a circle, called \emph{chord diagram} (see \autoref{fig_circle_graph} below for an example). 

\vspace{-0.6cm}
\begin{figure}[ht]
\begin{center}
\bigskip
\begin{tikzpicture}[thick,scale=0.65]
\tikzstyle{sommet}=[circle, draw, fill=black, inner sep=0pt, minimum width=4pt]
\tikzstyle{newsommet}=[circle, draw, fill=green!50!black, inner sep=0pt, minimum width=4pt]

\begin{scope}[xshift=-8cm,yshift=0cm,rotate=270]

\node[] (5) at (90+360/5:1.3) {} ;
\draw[] (5) node[sommet]{};
\node[] (55) at (90+360/5:1.7) {$b$};

\node[] (1) at (90+2*360/5:1.3) {} ;
\draw[] (1) node[sommet]{};
\node[] (11) at (90+2*360/5:1.7) {$e$};

\node[] (3) at (90+3*360/5:1.3) {} ;
\draw[] (3) node[sommet]{};
\node[] (33) at (90+3*360/5:1.7) {$d$};

\node[] (7) at (90+4*360/5:1.3) {} ;
\draw[] (7) node[sommet]{};
\node[] (77) at (90+4*360/5:1.7) {$c$};

\node[] (9) at (90:1.3) {} ;
\draw[] (9) node[sommet]{};
\node[] (99) at (90:1.7) {$a$};

\draw (1.center) -- (3.center) ;
\draw (1.center) -- (5.center) ;
\draw (3.center) -- (7.center) ;
\draw (5.center) -- (7.center) ;
\draw (5.center) -- (9.center) ;
\draw (7.center) -- (9.center) ;

\end{scope}

\begin{scope}[xshift=0cm,yshift=0cm,rotate=270,scale=0.8]

\draw[dotted] (0,0) circle [radius=2.2cm] ;

\node[] (a1) at (90-18+2*36:2.6) {$a_1$} ;
\node[] (a2) at (90-18+9*36:2.6) {$a_2$} ;
\node[] (aa1) at (90-18+2*36:2.2) {};
\node[] (aa2) at (90-18+9*36:2.2) {};
\draw[-] (aa1.center) -- (aa2.center) ;

\node[] (b1) at (90-18+36:2.6) {$b_1$} ;
\node[] (b2) at (90-18+7*36:2.6) {$b_2$} ;
\node[] (bb1) at (90-18+36:2.2) {};
\node[] (bb2) at (90-18+7*36:2.2) {};
\draw[-] (bb1.center) -- (bb2.center) ;

\node[] (c1) at (90-18:2.6) {$c_1$} ;
\node[] (c2) at (90-18+4*36:2.6) {$c_2$} ;
\node[] (cc1) at (90-18:2.2) {};
\node[] (cc2) at (90-18+4*36:2.2) {};
\draw[-] (cc1.center) -- (cc2.center) ;

\node[] (d1) at (90-18+3*36:2.6) {$d_2$} ;
\node[] (d2) at (90-18+6*36:2.6) {$d_1$} ;
\node[] (dd1) at (90-18+3*36:2.2) {};
\node[] (dd2) at (90-18+6*36:2.2) {};
\draw[-] (dd1.center) -- (dd2.center) ;

\node[] (e1) at (90-18+5*36:2.6) {$e_2$} ;
\node[] (e2) at (90-18+8*36:2.6) {$e_1$} ;
\node[] (ee1) at (90-18+5*36:2.2) {};
\node[] (ee2) at (90-18+8*36:2.2) {};
\draw[-] (ee1.center) -- (ee2.center) ;

\end{scope}

\end{tikzpicture}
\end{center}
\vspace{-0.2cm}
\caption{The House graph on the left and its chord diagram on the right.
\label{fig_circle_graph} 
}
\end{figure}
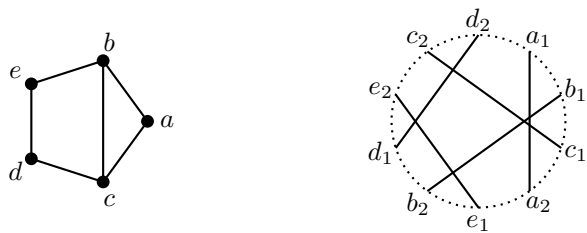

Since their introduction in the early 70's, circle graphs have been extensively studied. They were first defined by Even and Itai~\cite{EvenI71Queues}, who proved that the minimum number of parallel stacks needed to sort a permutation equals the chromatic number of a circle graph. 
Independently, Bouchet~\cite{Bouchet72Caracterisation} studied \emph{alternance graphs} to provide an algorithmic solution of the Gauss problem on self-intersection curves in the plane. 
A double occurence word on a set $L$ of letters is a word containing exactly two copies of every letter of $L$. An alternance graph is defined by a double occurence word on its set of vertices such that two vertices are adjacent if and only if their occurrences alternate in the word (see \autoref{sub_chord_diagram}). It is easy to see that alternance graphs are exactly circle graphs.
A first characterization of circle graphs was proposed by Fournier~\cite{Fournier78UneCaracterisation} in terms of ordered sets, while De Fraysseix~\cite{deFraysseix84Acharacterization} characterized circle graphs as intersection graphs of co-cyclic paths. 
As confirmed by recent results, circle graphs play a very important role in the theory of \emph{vertex minors} (see~\cite{Oum05Rankwidth}). A \emph{local-complementation} consists in replacing in a graph the neighbourhood of a vertex by its complement graph. A vertex-minor of a graph $G$ is a graph $H$ that can be obtained from $G$ by a series of local complementations and vertex deletions. It is easy to see from a chord diagram that every vertex-minor of a circle graph is itself a vertex minor (this was first noticed by Kotzig~\cite{Kotzig77Quelques} in a series of seminars). Bouchet~\cite{Bouchet94Circle} showed that a graph is a circle graph if and only if it excludes one of the three graphs of \autoref{fig_obstructions} as a vertex-minor. It is conjectured that the vertex-minor relation forms a well-quasi-ordering (see~\cite{Oum05Rankwidth}). To tackle this question, the \emph{rank-width} parameter has been introduced as the analog, for vertex-minor relation, of the tree-width for the minor relation~\cite{RoberstonS90Treewidth}. It turns out that circle graphs play the same role for the rank-width parameter, as planar graphs for treewidth. Indeed Geelen \etal~\cite{GeelenKMcCW23TheGrid} recently proved an analog of the grid minor exclusion theorem~\cite{RoberstonS86Excluding}: a graph has bounded rank-width if and only if it excludes as a vertex minor a sufficiently large comparability grid\footnote{For a positive integer $n$, the $n \times n$-comparability grid is the graph with vertex set
$\{(i, j) \mid i, j \in \{1, 2, . . . , n\}\}$ such that vertices $(i, j)$ and $(i', j')$ are adjacent
if either $i \leq i'$ and $j \leq j'$, or $i \geq i'$ and $j \geq j'$.}. In the same way as every planar graph is a minor of a large enough grid, every circle graph is a vertex minor of a sufficiently large comparability grid.

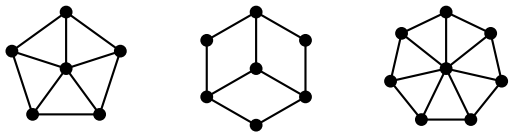
\begin{figure}[ht]
\begin{center}
\bigskip
\begin{tikzpicture}[thick,scale=0.5]
\tikzstyle{sommet}=[circle, draw, fill=black, inner sep=0pt, minimum width=4pt]
\tikzstyle{newsommet}=[circle, draw, fill=green!50!black, inner sep=0pt, minimum width=4pt]

\begin{scope}[xshift=-5cm,yshift=0cm]

\node[] (c) at (0:0) {} ;
\draw[] (c) node[sommet]{};
\foreach \i in {1,...,5}{
	\node[] (\i) at (90+\i*360/5:1.5) {} ;
	\draw[] (\i) node[sommet]{};
	\draw (c.center) -- (\i.center) ;
}

\draw (1.center) -- (2.center) ;
\draw (2.center) -- (3.center) ;
\draw (3.center) -- (4.center) ;
\draw (4.center) -- (5.center) ;
\draw (5.center) -- (1.center) ;
\end{scope}

\begin{scope}[xshift=0cm,yshift=0cm]

\node[] (c) at (0:0) {} ;
\draw[] (c) node[sommet]{};
\foreach \i in {1,...,6}{
	\node[] (\i) at (90+\i*360/6:1.5) {} ;
	\draw[] (\i) node[sommet]{};
}

\draw (1.center) -- (2.center) ;
\draw (2.center) -- (3.center) ;
\draw (3.center) -- (4.center) ;
\draw (4.center) -- (5.center) ;
\draw (5.center) -- (6.center) ;
\draw (6.center) -- (1.center) ;

\draw (c.center) -- (2.center) ;
\draw (c.center) -- (4.center) ;
\draw (c.center) -- (6.center) ;

\end{scope}
\begin{scope}[xshift=5cm,yshift=0cm]

\node[] (c) at (0:0) {} ;
\draw[] (c) node[sommet]{};
\foreach \i in {1,...,7}{
	\node[] (\i) at (90+\i*360/7:1.5) {} ;
	\draw[] (\i) node[sommet]{};
	\draw (c.center) -- (\i.center) ;
}

\draw (1.center) -- (2.center) ;
\draw (2.center) -- (3.center) ;
\draw (3.center) -- (4.center) ;
\draw (4.center) -- (5.center) ;
\draw (5.center) -- (6.center) ;
\draw (6.center) -- (7.center) ;
\draw (7.center) -- (1.center) ;
\end{scope}

\end{tikzpicture}
\end{center}
\caption{The vertex minor obstructions of circle graphs: $W_5$, the wheel of 5 vertices, $F_7$ and $W_7$.
\label{fig_obstructions} 
}
\end{figure}

As pointed out in~\cite{Bouchet87Reducing}, neither Fournier's  nor de Fraysseix's characterization yield a polynomial-time recognition algorithm for circle graph. It took about 15 years for the first recognition algorithm to appear. Actually, three polynomial-time recognition algorithms were independently obtained in the mid 80's by Naji~\cite{Naji85Reconnaissance}, by Gabor \etal~\cite{GaborHS89Recognizing}, and by Bouchet~\cite{Bouchet87Reducing}, respectively. Each of these three algorithms involves the \emph{split decomposition} of graphs, introduced by Cunningham and Edmonds~\cite{CunnighamE80Acombinatorial}. Indeed, they all reduce the circle graph recognition problem to \emph{prime} graphs, that is, graphs that are indecomposable with respect to the split decomposition. 
The $O(n^7)$-time algorithm of Naji~\cite{Naji85Reconnaissance} relies on the resolution of a system of linear equations that characterizes prime circle graphs. 
Bouchet's algorithm~\cite{Bouchet87Reducing} runs in $O(n^5)$-time and uses the property that every prime circle graph on $n$ vertices contains as a vertex minor a prime circle graph on $n-1$ vertices. The complexity of the algorithm of Gabor \etal~is $O(n^3)$. It relies on the property that prime circle graphs have a unique chord diagram, a property that was formally proved only later (see \cite{Bouchet87Reducing,Courcelle08Circle,GioanPTC14Circle}). The complexity of the circle graph recognition was later improved  down to $O(n^2)$-time by Spinrad~\cite{Spinrad94Recognition}. This was possible by the use of a novel $O(n^2)$-time algorithm to compute the split decomposition of a graph due to Ma and Spinrad~\cite{Mas94QuadraticSplit}. We observe that the linear-time split decomposition algorithm later obtained by Dahlhaus~\cite{Dahlhaus94SplitConf,Dahlaus00SplitJournal} has no impact on the complexity of Spinrad's circle graph recognition algorithm. The reason is that Spinrad's algorithm is a vertex-incremental algorithm, while Dahlhaus' is not.
Finally, until very recently no subquadratic time circle graph recognition algorithm was known. Gioan \etal~\cite{GioanPTC14Split,GioanPTC14Circle} broke the quadratic barrier by designing a novel almost linear-time split decomposition algorithm and showed how to apply that algorithm for the circle graph recognition problem. 
Their algorithm uses an important property of the last vertex of a Lexicographic Breadth-First-Search (LexBFS)  \cite{RoseTL76Algorithmic} with respect to the split-decomposition. In the case of circle graphs that property translates to the existence of a chord diagram where the neighbourhood of the last LexBFS vertex appears \emph{consecutively} (see \autoref{lem_good_vertex}, so-called \emph{good vertex lemma}). This allows Gioan \etal{} to design a LexBFS-based algorithm that incrementally updates the split decomposition and the corresponding representation of circle graphs. The hurdle to linear time complexity in the algorithm of Gioan \etal~is the use of the \emph{union-find} data-structure \cite{GallerF64AnImproved,Tarjan75Efficiency} to update the parent-children relationships in the modified split decomposition. The resulting time complexity is $O((n+m)\cdot \alpha(n,m))$, where $\alpha(\cdot)$ is the slowly growing inverse of the Ackermann function.

\medskip
\noindent
\textbf{Our result.} Gioan \etal~\cite{GioanPTC14Circle} left open the question of the existence of a linear-time circle graph recognition algorithm. 
To resolve that question,  we show that the consecutive property of the neighbourhood of the last vertex of a LexBFS ordering 
allows us to use a PC-tree data-structure~\cite{ShihH99Anew} (see also \cite{HsuMcC03PCTrees,FinkPR23Experimental}) rather than the union-find data-structure. As a consequence, we shave the $\alpha(n,m)$ factor in the time complexity of Gioan \etal's algorithm, yielding the first linear-time circle graph recognition algorithm.


An intriguing question that remains open is whether the LexBFS incremental split decomposition algorithm could also be turned into a linear time algorithm. But to circumvent the use of the union-find data-structure, in that case, we are still missing a generalization of good vertex lemma providing a consecutive property.

The linear time circle graph recognition algorithm has implications for other problems, where linear time was achieved only under the assumption that a representation of the input is provided.  This includes 
the isomorphism problem for circle graphs~\cite{KaliszKZ22} as well as the partial representation extension problem~\cite{BrucknerRS24Extending}, where one is given a chord diagram for an induced subgraph of the input graph and seeks to extend it to a representation of the full graph.  In addition, it may have pave the way for a new linear-time recognition algorithm of circular-arc graphs that follows Hsu's approach~\cite{Hsu95}.

\medskip
\noindent
\textbf{Organization of the paper.} \autoref{sec_prelim} introduces the basic definitons and presents the model of \emph{graph labelled trees} (GLT) proposed by \cite{GioanP12Split,GioanPTC14Split} to represent the split decomposition of a graph.
\autoref{sec_circle} is dedicated to circle graphs, their chord diagrams, the data-structures we use to store a circle graph.  The GLT representing the split decomposition tree is implemented by a PC-tree~\cite{HsuMcC03PCTrees,FinkPR23Experimental}, and the chord diagrams representing the prime nodes are encoded by Consistent Symmetric Cycles~\cite{GioanPTC14Circle}.  \autoref{sec:lexbfs} describes LexBFS and corresponding properties of circle graphs, including an alternative proof of the \emph{good vertex lemma}. The original proof of Gioan et al.~\cite{GioanPTC14Circle} was two-steps: it first proved the property for prime circle graphs and then the general case follows as a consequence of the recognition algorithm. Our proof is independent of the recognition algorithm. Finally, in~\autoref{sec_algo}, we present how to adapt the circle graph recognition of Gioan et al.~\cite{GioanPTC14Circle} to the new data-structure. The algorithm checks if a vertex ending a LexBFS can be inserted in the split PC-tree of a circle graph and if so updates the split PC-tree representation. Along with the correctness proof, we provide an amortized time complexity analysis based on the method of potentials~\cite{SleatorT85} (see also~\cite[Ch. 17.3]{cormen01introduction}). Interestingly, we observe that this complexity analysis also applies to the original algorithm~\cite{GioanPTC14Circle}, and to the split decomposition algorithm~\cite{GioanPTC14Split} as well, and simplifies the complexity therein.

\section{Preliminaries and split decomposition}
\label{sec_prelim}

\subsection{Basic definitions}
\label{sub_basic_def}

A \emph{word} over an alphabet $\Sigma$ is a finite sequence of letters of $\Sigma$. If $A$ is a word over $\Sigma$, then $A^r$ denotes its \emph{reversal} (the ordering between the letters of $A$ has been reversed). The concatenation between two words $A$ and $B$ over $\Sigma$ is a word over $\Sigma$ denoted $AB$. A subsequence $F$ of consecutive letters in a word $A$ is called a \emph{factor} of $A$.

A \emph{circular word} $C$ over $\Sigma$ is a circular sequence of letters of $\Sigma$. It can be represented by a word by considering that the first letter follows the last letter. Observe that such a representation fixes an arbitrary first letter. So if a circular word $C$ is represented by the word $AB$, then $BA$ also represents $C$.
To denote that equivalence, we write $C\sim AB\sim BA$. Observe that, if $C$ a circular word represented by the word $A$, that is $C\sim A$, then the reversal of $C$ is $C^r\sim A^r$. The notion of factor naturally extends to circular words: $F$ is a factor of $C$ if there exists a word $A$ such that $C\sim A$ and $F$ is a factor of $A$. 

Unless specified, we will assume that every graph $G=(V,E)$, on vertex set $V$ and edge set $E$, is connected. The neighbourhood of a vertex $x$ in $G$ will be denoted $N_{G}(x)$ or simply~$N(x)$ if the graph is clear from the context. A \emph{clique} is a graph where every vertex is adjacent to every other vertex. A \emph{star} is a tree with one universal vertex. For a graph $G=(V,E)$ and a vertex $x$ not belonging to $V$, we let denote $G'=G+x$~\footnote{Though the neighbourhood  of $x$ is not specified in the notation $G+x$, it will always be clear from the context.} the graph obtained by adding $x$ to $V$ and all the edges between $x$ and $N_{G'}(x)$ to $E$. Given a subset $S$ of vertices of $V$, we let $G[S]$ denote the subgraph induced by $S$. A \emph{vertex ordering} $\sigma$ is a total order on the vertices of a graph $G$. If $x$ is smaller than $y$ in $\sigma$, we write $x<_{\sigma} y$. For a subset $S$ of vertices, we let $\sigma[S]$ denote the subordering of $\sigma$ induced by the vertices of $S$, that is for every $x,y\in S$, $x<_{\sigma[S]} y$ if and only if $x<_{\sigma} y$.

\subsection{Split decomposition and graph labelled trees}
\label{sub_split_GLT}

A bipartition $(A,B)$ of a set $V$ is \emph{non-trivial} if $|A|>1$, $|B|>1$.

\begin{definition}[Split of a graph]
  A \emph{split} of a graph $G=(V,E)$ is a non-trivial bipartition $(A,B)$ of its vertex set $V$ with two subsets $A'\subseteq A$, $B'\subseteq B$, called \emph{frontiers} of the split, such that for every $x\in A$ and $y\in B$, it holds that $xy\in E$ if and only if $x\in A'$ and $y\in B'$. 
\end{definition}

Splitting a graph $G$ with respect to a split $(A,B)$ of $G$ returns two graphs $G_A$ and $G_B$ obtained from~$G$ by contracting $B$ into a vertex $x_A$ and $A$ into a vertex $x_B$ (the neighbourhood of $x_A$ in $G_A$ is $A'$ and of $x_B$ in $G_B$ is $B'$), respectively.
Let  $G=(V,E)$ and $G'=(V',E')$ be two graphs and let $x\in V$ and $x'\in V'$ be two vertices. The \emph{join} between $G$ and $G'$ with respect to $x$ and $x'$, denoted $(G,x)\otimes(G',x')$, is the graph on vertex set $(V\cup V')\setminus \{x,x'\}$ such that two vertices $y$ and $z$ are adjacent if and only if $yz\in E$, or $yz\in E'$, or $y\in N_G(x)$ and $z\in N_{G'}(x')$. So the join and the splitting are reverse operations for large enough graphs.
Observe that if  $G$ is the single vertex graph, then $(G,x)\otimes(G',x')$ is isomorphic to $G[V\setminus\{x'\}]$, and if $G$ is the clique of size $2$, then $(G,x)\otimes(G',x')$ is isomorphic to $G'$
Otherwise, if $G$ and $G'$ are connected, then $(V\setminus\{x\},V'\setminus \{x'\})$ is a split of $(G,x)\otimes(G',x')$. 
A graph $G$ is \emph{degenerate} if every non-trivial bipartition of $V$ is a split of $G$. It is known that if $G$ is degenerate, then $G$ is either a \emph{clique} or a \emph{star}. A graph $G$ is \emph{prime} if it does not contain any split.

\begin{definition}[Graph labelled tree]
A \emph{graph labelled tree} (\GLT) is a pair $(T,\F)$, where $T$ is a tree and $\mathcal{F}$ a set of graphs, called \emph{label graphs}, such that each node $u$ of $T$ is labelled by the graph $G(u)\in \F$ and is further equipped with a bijection~$\rho_u$ that bijectively maps the edges incident to~$u$ to the vertices of~$G(u)$.
\end{definition}

Let $(T,\F)$ be a \GLT{} and let $u$ be a node of $(T,\F)$. We say that the node $u$ is \emph{prime} if $G(u)$ is a prime graph and that $u$ is \emph{degenerate} if $G(u)$ is degenerate.
We let $V(u)$ denote the vertices of $G(u)$, hereafter called \emph{marker vertices}, and $E(u)$ denote the edges of $G(u)$, hereafter called \emph{label edges}.  Since a leaf has only one incident edge, we may abusively consider leaves as marker vertices (namely the leaf~$u$ is identified with the marker~$\rho_u(e)$, of the unique edge~$e$ incident to~$u$). If $e=uv$ is a tree-edge of $T$, then the marker vertices $\rho_u(e)$ and $\rho_v(e)$ are \emph{opposite} marker vertices called the \emph{extremities} of $e$. Observe that every marker vertex is the extremity of a tree-edge of $T$. 

A graph labelled tree is designed to represent a graph in a compact manner. To explain how, we need to introduce the notion of accessibility, which can be seen as an extension of the adjacency. Two marker vertices $q$ and $q'$ of distinct nodes are \emph{accessible} from one another if there exists a sequence of marker vertices $\langle q=q_1,\dots , q_{2k}=q'\rangle$ of even length such that~$q_i,q_{i+1}$ are the extremities of a tree-edge if~$i$ is odd and are adjacent marker vertices in some label graph if~$i$ is even.

Let $e=uu'$ be a tree-edge of the \GLT{} $(T,\F)$. Consider the marker vertex $q\in V(u)$ that is the extremity $e$. Then, we let $L(q)$ denote the set of leaves of the subtree of $T-e$ not containing~$u$. Among $L(q)$, we distinguish the subset $A(q)$ of leaves accessible from $q$. 

\begin{definition}[Accessiblity graph]
The \emph{accessibility graph} $\Gr(T,\F)$ of a \GLT{} $(T,\F)$ is the graph whose vertices are the leaves of $T$ and such that two vertices are adjacent if and only if the corresponding leaves are accessible from one another. 
\end{definition}

Let $q'\in V(u')$ be the extremity of $e$ distinct from~$q$. We observe that if $e$ is not incident to a leaf of $T$, then the bipartition $(L(q),L(q'))$ is a split of $\Gr(T,\F)$
and the frontiers of that split are $A(q)$ and $A(q')$. 
The \emph{node-join} of two non-leaf nodes $u$ and $u'$ in $(T,\F)$ returns the \GLT{} $(T',\F')$ obtained by contracting the tree-edge $e=uu'$, labelling the resulting node~$v$ by the graph $G(v)=(G(u),q)\otimes (G(u'),q')$, and letting every marker vertex $q\in V(v)$ be the extremity of the same tree-edge as in $(T,\F)$.
The \emph{node-split} operation of a \GLT{} is the reverse of the node-join operation. Suppose that $(A,B)$ is a split of the graph $G(u)$ for some node $u$ of $(T,\F)$. The node-split of $u$ with respect to $(A,B)$ returns the \GLT{} $(T',\F')$ where node $u$ has been replaced by two adjacent nodes $u_A$ and $u_B$  labelled by $G(u_A)=G(u)[A\cup\{q_A\}]$ for some vertex $q_A$ of the frontier $A'$ and $G(u_B)=G(u)[B\cup\{q_B\}]$ for some vertex $q_A$ of the frontier $A'$, respectively. The marker vertices $q_A\in V(u_A)$ and $q_A\in V(u_B)$ are made the extremities of the tree-edge $u_Au_B$. Every other marker vertex is an extremity of the same tree-edge as in $(T,\F)$.

\begin{theorem}[\cite{CunnighamE80Acombinatorial,GioanPTC14Split}]
For any connected graph $G$, there exists a unique \GLT{}, denoted $\ST(G)$ and called the \emph{split-tree} of $G$, whose labels are either prime or degenerate, having a minimum number of nodes and such that $\Gr(\ST(G))=G$.
\end{theorem}

\section{Circle graphs and split PC-trees}
\label{sec_circle}

\subsection{Chord diagrams and circle graphs}
\label{sub_chord_diagram}

A \emph{chord} on a circle is defined by a pair $c=\{c_1,c_2\}$ of distinct points of the circle, called \emph{endpoints}. Let $\chi$ be a set of chords no pair of which share a common endpoint. An \emph{expanded chord diagram} on $\chi$
is naturally represented by a circular word $\dot{D}$ over the alphabet $\V=\bigcup_{c\in \chi}\{c_1,c_2\}$. A \emph{chord diagram} $D$ is obtained from the expanded chord diagram $\dot{D}$ by replacing every endpoint $c_i\in \V$ ($i\in\{1,2\}$) by the corresponding chord $c\in \chi$. Observe that a chord diagram $D$ is a \emph{double occurence circular word} over the alphabet $\chi$, that is every chord of $\chi$ appears exactly twice in $D$. We may abusively say that each occurence of a chord~$c$ in $D$ is an endpoints of $c$. A subset $S$ of chords of $\chi$ is \emph{consecutive} in $D$ if $D$ contains a factor~$F$ such that for every chord $c\in S$, $|c\cap F|=1$ and for every chord $c'\notin S$, $c'\cap F=\emptyset$. The first and the last chord of the factor $F$ are called the \emph{bookends} of $F$ (or of $S$).
Let $a$ and $b$ be two endpoints of the chord diagram $D$. Then $D(a,b)$ is the factor of $D$ containing the set of endpoints comprised between $a$ and $b$ (not containing $a$ and $b$), while $D(b,a)$ contains those comprised between $b$ and $a$. In other words, we have that $D\sim aD(a,b)bD(b,a)$. Observe that $D(a,b)^r=D^r(b,a)$ and that $D(b,a)^r=D^r(a,b)$.

Let $D$ and $D'$ be chord diagrams on $\chi$ and $\chi'$, respectively, and let $x\in \chi$ and $x'\in \chi'$. We define the \emph{circle-join} operation between $D$ and $D'$ with respect to $x=\{x_1,x_2\}$ and $x'=\{x_1',x_2'\}$ as follows:
\[(D,x)\odot (D',x')\sim D(x_1,x_2)D'(x'_1,x'_2)D(x_2,x_1)D'(x'_2,x'_1)\]
The result is a chord diagram on the set $(\chi\cup \chi')\setminus\{x,x'\}$ (see \autoref{fig_node_join}).
We observe that the circle-join is not a commutative operation: $(D,x)\odot (D',x')\neq (D',x')\odot (D,x)$.

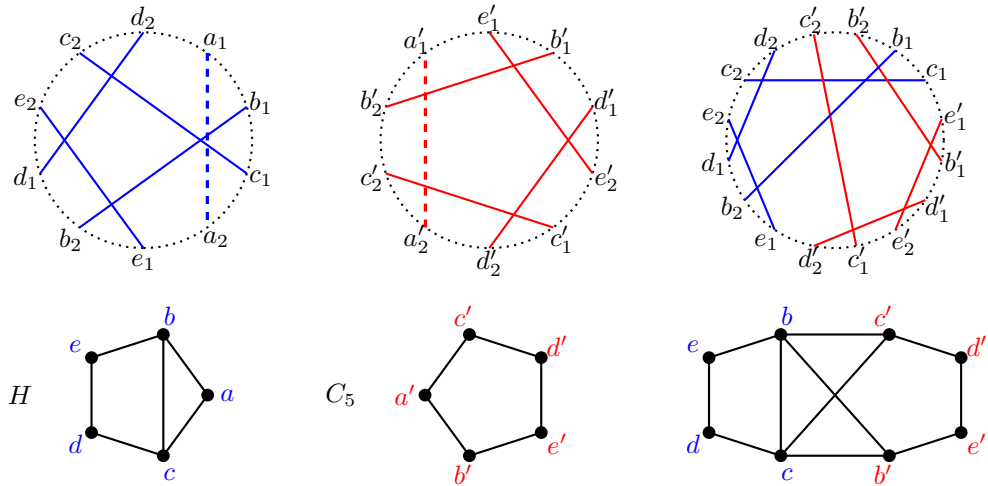
\begin{figure}[ht]
\begin{center}
\bigskip
\begin{tikzpicture}[thick,scale=0.65]
\tikzstyle{sommet}=[circle, draw, fill=black, inner sep=0pt, minimum width=4pt]
\tikzstyle{newsommet}=[circle, draw, fill=green!50!black, inner sep=0pt, minimum width=4pt]

\begin{scope}[xshift=0cm,yshift=-5.2cm,rotate=270]

\node[] (H) at (270:2.5) {$H$};

\node[] (5) at (90+360/5:1.3) {} ;
\draw[] (5) node[sommet]{};
\node[blue] (55) at (90+360/5:1.7) {$b$};

\node[] (1) at (90+2*360/5:1.3) {} ;
\draw[] (1) node[sommet]{};
\node[blue] (11) at (90+2*360/5:1.7) {$e$};

\node[] (3) at (90+3*360/5:1.3) {} ;
\draw[] (3) node[sommet]{};
\node[blue] (33) at (90+3*360/5:1.7) {$d$};

\node[] (7) at (90+4*360/5:1.3) {} ;
\draw[] (7) node[sommet]{};
\node[blue] (77) at (90+4*360/5:1.7) {$c$};

\node[] (9) at (90:1.3) {} ;
\draw[] (9) node[sommet]{};
\node[blue] (99) at (90:1.7) {$a$};

\draw (1.center) -- (3.center) ;
\draw (1.center) -- (5.center) ;
\draw (3.center) -- (7.center) ;
\draw (5.center) -- (7.center) ;
\draw (5.center) -- (9.center) ;
\draw (7.center) -- (9.center) ;

\end{scope}

\begin{scope}[xshift=0cm,yshift=0cm,rotate=270]

\draw[dotted] (0,0) circle [radius=2.2cm] ;

\node[] (a1) at (90-18+2*36:2.5) {$a_1$} ;
\node[] (a2) at (90-18+9*36:2.5) {$a_2$} ;
\node[] (aa1) at (90-18+2*36:2.2) {};
\node[] (aa2) at (90-18+9*36:2.2) {};
\draw[blue,very thick,dashed] (aa1.center) -- (aa2.center) ;

\node[] (b1) at (90-18+36:2.5) {$b_1$} ;
\node[] (b2) at (90-18+7*36:2.5) {$b_2$} ;
\node[] (bb1) at (90-18+36:2.2) {};
\node[] (bb2) at (90-18+7*36:2.2) {};
\draw[blue,-] (bb1.center) -- (bb2.center) ;

\node[] (c1) at (90-18:2.5) {$c_1$} ;
\node[] (c2) at (90-18+4*36:2.5) {$c_2$} ;
\node[] (cc1) at (90-18:2.2) {};
\node[] (cc2) at (90-18+4*36:2.2) {};
\draw[blue,-] (cc1.center) -- (cc2.center) ;

\node[] (d1) at (90-18+3*36:2.5) {$d_2$} ;
\node[] (d2) at (90-18+6*36:2.5) {$d_1$} ;
\node[] (dd1) at (90-18+3*36:2.2) {};
\node[] (dd2) at (90-18+6*36:2.2) {};
\draw[blue,-] (dd1.center) -- (dd2.center) ;

\node[] (e1) at (90-18+5*36:2.5) {$e_2$} ;
\node[] (e2) at (90-18+8*36:2.5) {$e_1$} ;
\node[] (ee1) at (90-18+5*36:2.2) {};
\node[] (ee2) at (90-18+8*36:2.2) {};
\draw[blue,-] (ee1.center) -- (ee2.center) ;

\end{scope}

\begin{scope}[xshift=7cm,yshift=-5.2cm,rotate=90]

\node[] (C) at (90:3) {$C_5$};

\node[] (5) at (90+360/5:1.3) {} ;
\draw[] (5) node[sommet]{};
\node[red] (55) at (90+360/5:1.7) {$b'$};

\node[] (1) at (90+2*360/5:1.3) {} ;
\draw[] (1) node[sommet]{};
\node[red] (11) at (90+2*360/5:1.7) {$e'$};

\node[] (3) at (90+3*360/5:1.3) {} ;
\draw[] (3) node[sommet]{};
\node[red] (33) at (90+3*360/5:1.7) {$d'$};

\node[] (7) at (90+4*360/5:1.3) {} ;
\draw[] (7) node[sommet]{};
\node[red] (77) at (90+4*360/5:1.7) {$c'$};

\node[] (9) at (90:1.3) {} ;
\draw[] (9) node[sommet]{};
\node[red] (99) at (90:1.7) {$a'$};

\draw (1.center) -- (3.center) ;
\draw (1.center) -- (5.center) ;
\draw (3.center) -- (7.center) ;
\draw (5.center) -- (9.center) ;
\draw (7.center) -- (9.center) ;

\end{scope}

\begin{scope}[xshift=7cm,yshift=0cm,rotate=90]

\draw[dotted] (0,0) circle [radius=2.2cm] ;

\node[] (a1) at (90-18+2*36:2.5) {$a'_2$} ;
\node[] (a2) at (90-18+9*36:2.5) {$a'_1$} ;
\node[] (aa1) at (90-18+2*36:2.2) {};
\node[] (aa2) at (90-18+9*36:2.2) {};
\draw[red,very thick,dashed] (aa1.center) -- (aa2.center) ;

\node[] (b1) at (90-18:2.5) {$b'_2$} ;
\node[] (b2) at (90-18+7*36:2.5) {$b'_1$} ;
\node[] (bb1) at (90-18:2.2) {};
\node[] (bb2) at (90-18+7*36:2.2) {};
\draw[red,-] (bb1.center) -- (bb2.center) ;

\node[] (c1) at (90-18+36:2.5) {$c'_2$} ;
\node[] (c2) at (90-18+4*36:2.5) {$c'_1$} ;
\node[] (cc1) at (90-18+36:2.2) {};
\node[] (cc2) at (90-18+4*36:2.2) {};
\draw[red,-] (cc1.center) -- (cc2.center) ;

\node[] (d1) at (90-18+3*36:2.5) {$d'_2$} ;
\node[] (d2) at (90-18+6*36:2.5) {$d'_1$} ;
\node[] (dd1) at (90-18+3*36:2.2) {};
\node[] (dd2) at (90-18+6*36:2.2) {};
\draw[red,-] (dd1.center) -- (dd2.center) ;

\node[] (e1) at (90-18+5*36:2.5) {$e'_2$} ;
\node[] (e2) at (90-18+8*36:2.5) {$e'_1$} ;
\node[] (ee1) at (90-18+5*36:2.2) {};
\node[] (ee2) at (90-18+8*36:2.2) {};
\draw[red,-] (ee1.center) -- (ee2.center) ;

\end{scope}

\begin{scope}[xshift=14cm,yshift=0cm,rotate=0]

\draw[dotted] (0,0) circle [radius=2.2cm] ;

\node[] (d1) at (180+11.25:2.5) {$d_1$} ;
\node[] (d2) at (180+11.25-3*22.5:2.5) {$d_2$} ;
\node[] (dd1) at (180+11.25:2.2) {};
\node[] (dd2) at (180+11.25-3*22.5:2.2) {};
\draw[blue,-] (dd1.center) -- (dd2.center) ;

\node[] (e1) at (180+11.25-22.5:2.5) {$e_2$} ;
\node[] (e2) at (180+11.25-14*22.5:2.5) {$e_1$} ;
\node[] (ee1) at (180+11.25-22.5:2.2) {};
\node[] (ee2) at (180+11.25-14*22.5:2.2) {};
\draw[blue,-] (ee1.center) -- (ee2.center) ;

\node[] (c1) at (180+11.25-2*22.5:2.5) {$c_2$} ;
\node[] (c2) at (180+11.25-7*22.5:2.5) {$c_1$} ;
\node[] (cc1) at (180+11.25-2*+22.5:2.2) {};
\node[] (cc2) at (180+11.25-7*22.5:2.2) {};
\draw[blue,-] (cc1.center) -- (cc2.center) ;

\node[] (c'1) at (180+11.25-4*22.5:2.5) {$c'_2$} ;
\node[] (c'2) at (180+11.25-12*22.5:2.5) {$c'_1$} ;
\node[] (cc'1) at (180+11.25-4*22.5:2.2) {};
\node[] (cc'2) at (180+11.25-12*22.5:2.2) {};
\draw[red,-] (cc'1.center) -- (cc'2.center) ;

\node[] (b'1) at (180+11.25-5*22.5:2.5) {$b'_2$} ;
\node[] (b''2) at (180+11.25-9*22.5:2.5) {$b'_1$} ;
\node[] (bb'1) at (180+11.25-5*22.5:2.2) {};
\node[] (bb'2) at (180+11.25-9*22.5:2.2) {};
\draw[red,-] (bb'1.center) -- (bb'2.center) ;

\node[] (b1) at (180+11.25-6*22.5:2.5) {$b_1$} ;
\node[] (b2) at (180+11.25-15*22.5:2.5) {$b_2$} ;
\node[] (bb1) at (180+11.25-6*22.5:2.2) {};
\node[] (bb2) at (180+11.25-15*22.5:2.2) {};
\draw[blue,-] (bb1.center) -- (bb2.center) ;

\node[] (e'1) at (180+11.25-8*22.5:2.5) {$e'_1$} ;
\node[] (e'2) at (180+11.25-11*22.5:2.5) {$e'_2$} ;
\node[] (ee'1) at (180+11.25-8*22.5:2.2) {};
\node[] (ee'2) at (180+11.25-11*22.5:2.2) {};
\draw[red,-] (ee'1.center) -- (ee'2.center) ;

\node[] (d'1) at (180+11.25-10*22.5:2.5) {$d'_1$} ;
\node[] (d'2) at (180+11.25-13*22.5:2.5) {$d'_2$} ;
\node[] (dd'1) at (180+11.25-10*22.5:2.2) {};
\node[] (dd'2) at (180+11.25-13*22.5:2.2) {};
\draw[red,-] (dd'1.center) -- (dd'2.center) ;

\end{scope}

\begin{scope}[xshift=14cm,yshift=-5.2cm,rotate=0]

\begin{scope}[xshift=-1.5cm,yshift=0cm,rotate=270]

\node[] (5) at (90+360/5:1.3) {} ;
\draw[] (5) node[sommet]{};
\node[blue] (55) at (90+360/5:1.7) {$b$};

\node[] (1) at (90+2*360/5:1.3) {} ;
\draw[] (1) node[sommet]{};
\node[blue] (11) at (90+2*360/5:1.7) {$e$};

\node[] (3) at (90+3*360/5:1.3) {} ;
\draw[] (3) node[sommet]{};
\node[blue] (33) at (90+3*360/5:1.7) {$d$};

\node[] (7) at (90+4*360/5:1.3) {} ;
\draw[] (7) node[sommet]{};
\node[blue] (77) at (90+4*360/5:1.7) {$c$};

\draw (1.center) -- (3.center) ;
\draw (1.center) -- (5.center) ;
\draw (3.center) -- (7.center) ;
\draw (5.center) -- (7.center) ;

\end{scope}

\begin{scope}[xshift=1.5cm,yshift=0cm,rotate=90]

\node[] (5') at (90+360/5:1.3) {} ;
\draw[] (5') node[sommet]{};
\node[red] (55') at (90+360/5:1.7) {$b'$};

\node[] (1') at (90+2*360/5:1.3) {} ;
\draw[] (1') node[sommet]{};
\node[red] (11') at (90+2*360/5:1.7) {$e'$};

\node[] (3') at (90+3*360/5:1.3) {} ;
\draw[] (3') node[sommet]{};
\node[red] (33') at (90+3*360/5:1.7) {$d'$};

\node[] (7') at (90+4*360/5:1.3) {} ;
\draw[] (7') node[sommet]{};
\node[red] (77') at (90+4*360/5:1.7) {$c'$};

\draw (1'.center) -- (3'.center) ;
\draw (1'.center) -- (5'.center) ;
\draw (3'.center) -- (7'.center) ;

\end{scope}

\draw (5.center) -- (5'.center) ;
\draw (5.center) -- (7'.center) ;
\draw (7.center) -- (5'.center) ;
\draw (7.center) -- (7'.center) ;

\end{scope}

\end{tikzpicture}
\end{center}
\vspace{-0.2cm}
\caption{From left to right: the chord diagram $D_1$ of the house graph $H$, the chord diagram $D_2$ of the $C_5$ and $D_1(a)\odot D_2(a')=D_1(a_1,a_2)D_2(a'_1,a'_2)D_1(a_2,a_1)D_2(a'_2,a'_1)$, the chord diagram of the graph $G=(H,a)\otimes (C_5,a')$ (depicted on the right). Dotted chords represent the vertices (respectively $a$ and $a'$) on which the join is performed. In $G$, vertices $b$ and $c$ (the neighbors of $a$ in the House graph) are both adjacent to vertices $b'$ and $c'$ (the neighbors of $a'$ in the $C_5$).
\label{fig_node_join} 
}
\end{figure}

\begin{lemma}[{\cite[Lemma 3.3]{GioanPTC14Circle}}]
\label{lem_consecutive_join_diagram}
Let $D$ and $D'$ be chord diagrams on the sets $V$ and $V'$ of chords, respectively. Let $S\subset V$ and $S'\subset V'$ be consecutive sets of chords in their respective chord diagrams such that $1<|S|<|V|$ and $1<|S'|<|V'|$. If $x$ and $x'$ are bookends of $S$ and $S'$, respectively, then $X=(S\setminus\{x\})\cup (S'\setminus\{x'\})$ is consecutive in (at least) one of the following chord diagrams:
$$ (D,x)\odot (D',x'),~~~~ (D',x')\odot (D,x),~~~~ (D,x)\odot (D'^r,x'),~~~~ (D'^r,x')\odot (D,x).$$
Moreover, the bookends of $X$ are those of $S$ and $S'$ distinct from $x$ and $x'$.
\end{lemma}

A \emph{circle graph} $G=(V,E)$ is the intersection graph of a chord diagram $D$ on a set $\chi$ of chords of a circle. In other words, there exists a bijection $\rho:V\rightarrow \chi$ such that two vertices $x,y\in V$ are adjacent in $G$ if and only if the  chords $\rho(x)$ and $\rho(y)$ intersect in~$D$, that is if and only $D(x_1,x_2)$ contains one endpoint among $y_1$ and $y_2$ and $D(x_2,x_1)$ the other. We say that $D$ \emph{represents} or \emph{encodes} $G$. We observe that, in general, a circle graph is encoded by many different chord diagrams. For example, cliques and stars are circle graphs that are represented by many chord diagrams: every clique $G$ on vertex set $V$ is represented by $D\sim AA^r$ where $A$ is an arbitrary permutation of $V$; every star $G$ with center vertex $x$ is represented by $D\sim xAxA$ where $A$ is an arbitrary permutation of $V\setminus\{x\}$. 

\begin{observation} \label{obs_join_circle}
Let $G$ and $G'$ be two circle graphs represented by chord diagrams $D$ and $D'$,  respectively. Then, for any vertex $x$ of $G$ and any vertex $x'$ of $G'$, the chord diagrams $(D,x)\odot (D',x')$ and $(D',x')\odot (D,x)$ represent the circle graph $H=(G,x)\otimes (G',x')$.
\end{observation}

However, we have the following important property, announced in \cite{GaborHS89Recognizing}, partially proved in \cite{Bouchet87Reducing}, formally proved in \cite{Courcelle08Circle} (an alternative proof was given in \cite{GioanPTC14Circle}):

\begin{theorem}[\cite{GaborHS89Recognizing,Bouchet87Reducing,Courcelle08Circle,GioanPTC14Circle}]
\label{th_unique_prime}
A circle graph $G=(V,E)$ is a prime graph if and only if it has a unique (up to reversal) chord diagram.
\end{theorem}

Since degenerate graphs (cliques and stars) are circle graphs, a consequence of 
\autoref{obs_join_circle} is the following well-known characterization of circle graphs.

\begin{theorem}[\cite{GaborHS89Recognizing,Bouchet87Reducing,Naji85Reconnaissance}] \label{theo_circle_split}
A graph $G$ is a circle graph if and only if every prime node $u$ of $\ST(G)$ is labelled by a circle graph $G(u)$.
\end{theorem}


\subsection{Consistent symmetric cycles}

We describe the \emph{consistent symmetric cycle} data-structure (\CSC{}) introduced in~\cite{GioanPTC14Circle}  that represents chord diagrams and allows to efficiently perform the following operations: consecutive test (see \autoref{lem_consecutive_test}), consecutive preserving join (see~\autoref{lem_consecutive_join}) and vertex insertion (see~\autoref{lem_chord_insertion}).
Let $D$ be a chord diagram on the set $\chi$ of chords. The \emph{consistent symmetric cycle} (\CSC{}), denoted $\C(D)$, representing $D$ is a directed graph whose vertices are the chord endpoints of $\chi$, that is $\mu=\bigcup_{x\in\chi} \{x_1,x_2\}$. Every pair of consecutive endpoints in $D$ are linked by two symmetric arcs, hence forming a symmetric directed cycle. It follows that every endpoint has two out-neighbours, which we denote $+_{\C(D)}(\cdot)$ and $-_{\C(D)}(\cdot)$. For every chord $x$ we have the property $+_{\C(D)}(x_1)$ and $+_{\C(D)}(x_2)$ belong to one of the two connected components of $\C(D)\setminus\{x_1,x_2\}$ and $-_{\C(D)}(x_1)$ and $-_{\C(D)}(x_2)$ to the other.  To complete $\C(D)$, the two endpoints of every chord are linked with symmetric arcs (see~\autoref{fig_CSC}).

\begin{figure}[ht]
\begin{center}
\bigskip
\begin{tikzpicture}[thick,scale=0.6]
\tikzstyle{sommet}=[circle, draw, fill=black, inner sep=0pt, minimum width=4pt]
\tikzstyle{newsommet}=[circle, draw, fill=green!50!black, inner sep=0pt, minimum width=4pt]

\begin{scope}[xshift=-5cm,yshift=0cm,rotate=180]

\node[] (5) at (90+360/5:1.5) {} ;
\draw[] (5) node[sommet]{};
\node[] (55) at (90+360/5:1.9) {$b$};

\node[] (1) at (90+2*360/5:1.5) {} ;
\draw[] (1) node[sommet]{};
\node[] (11) at (90+2*360/5:1.9) {$e$};

\node[] (3) at (90+3*360/5:1.5) {} ;
\draw[] (3) node[sommet]{};
\node[] (33) at (90+3*360/5:1.9) {$d$};

\node[] (7) at (90+4*360/5:1.5) {} ;
\draw[] (7) node[sommet]{};
\node[] (77) at (90+4*360/5:1.9) {$c$};

\node[] (9) at (90:1.5) {} ;
\draw[] (9) node[sommet]{};
\node[] (99) at (90:1.9) {$a$};

\draw (1.center) -- (3.center) ;
\draw (1.center) -- (5.center) ;
\draw (3.center) -- (7.center) ;
\draw (5.center) -- (7.center) ;
\draw (5.center) -- (9.center) ;
\draw (7.center) -- (9.center) ;

\end{scope}

\begin{scope}[xshift=1.5cm,yshift=0cm,rotate=180]

\draw[dotted] (0,0) circle [radius=2.2cm] ;

\node[] (a1) at (90-18+2*36:2.6) {$a_1$} ;
\node[] (a2) at (90-18+9*36:2.6) {$a_2$} ;
\node[] (aa1) at (90-18+2*36:2.2) {};
\node[] (aa2) at (90-18+9*36:2.2) {};
\draw[-] (aa1.center) -- (aa2.center) ;

\node[] (b1) at (90-18+36:2.6) {$b_1$} ;
\node[] (b2) at (90-18+7*36:2.6) {$b_2$} ;
\node[] (bb1) at (90-18+36:2.2) {};
\node[] (bb2) at (90-18+7*36:2.2) {};
\draw[-] (bb1.center) -- (bb2.center) ;

\node[] (c1) at (90-18:2.6) {$c_1$} ;
\node[] (c2) at (90-18+4*36:2.6) {$c_2$} ;
\node[] (cc1) at (90-18:2.2) {};
\node[] (cc2) at (90-18+4*36:2.2) {};
\draw[-] (cc1.center) -- (cc2.center) ;

\node[] (d1) at (90-18+3*36:2.6) {$d_1$} ;
\node[] (d2) at (90-18+6*36:2.6) {$d_2$} ;
\node[] (dd1) at (90-18+3*36:2.2) {};
\node[] (dd2) at (90-18+6*36:2.2) {};
\draw[-] (dd1.center) -- (dd2.center) ;

\node[] (e1) at (90-18+5*36:2.6) {$e_1$} ;
\node[] (e2) at (90-18+8*36:2.6) {$e_2$} ;
\node[] (ee1) at (90-18+5*36:2.2) {};
\node[] (ee2) at (90-18+8*36:2.2) {};
\draw[-] (ee1.center) -- (ee2.center) ;

\end{scope}

\begin{scope}[scale=1.1,xshift=8cm,yshift=0cm,rotate=180]

\node[] (a1) at (90-18+2*36:2.6) {$a_1$} ;
\node[] (a2) at (90-18+9*36:2.6) {$a_2$} ;
\node[] (aa1) at (90-18+2*36:2.2) {};
\node[] (aa2) at (90-18+9*36:2.2) {};
\draw[<->] (aa1.center) -- (aa2.center) ;

\node[] (b1) at (90-18+36:2.6) {$b_1$} ;
\node[] (b2) at (90-18+7*36:2.6) {$b_2$} ;
\node[] (bb1) at (90-18+36:2.2) {};
\node[] (bb2) at (90-18+7*36:2.2) {};
\draw[<->] (bb1.center) -- (bb2.center) ;

\node[] (c1) at (90-18:2.6) {$c_1$} ;
\node[] (c2) at (90-18+4*36:2.6) {$c_2$} ;
\node[] (cc1) at (90-18:2.2) {};
\node[] (cc2) at (90-18+4*36:2.2) {};
\draw[<->] (cc1.center) -- (cc2.center) ;

\node[] (d1) at (90-18+3*36:2.6) {$d_1$} ;
\node[] (d2) at (90-18+6*36:2.6) {$d_2$} ;
\node[] (dd1) at (90-18+3*36:2.2) {};
\node[] (dd2) at (90-18+6*36:2.2) {};
\draw[<->] (dd1.center) -- (dd2.center) ;

\node[] (e1) at (90-18+5*36:2.6) {$e_1$} ;
\node[] (e2) at (90-18+8*36:2.6) {$e_2$} ;
\node[] (ee1) at (90-18+5*36:2.2) {};
\node[] (ee2) at (90-18+8*36:2.2) {};
\draw[<->] (ee1.center) -- (ee2.center) ;

\node[rotate=0*36] (p1) at (90+0*36:2.6) {\tiny $+$} ;
\node[rotate=0*36] (m1) at (90+0*36:2.3) {\tiny$+$} ;
\node[rotate=1*36] (p2) at (90+1*36:2.6) {\tiny $-$} ;
\node[rotate=1*36] (m2) at (90+1*36:2.3) {\tiny$+$} ;
\node[rotate=2*36] (p3) at (90+2*36:2.6) {\tiny $-$} ;
\node[rotate=2*36] (m3) at (90+2*36:2.3) {\tiny$+$} ;
\node[rotate=3*36+180] (p4) at (90+3*36:2.6) {\tiny $+$} ;
\node[rotate=3*36+180] (m4) at (90+3*36:2.3) {\tiny$-$} ;
\node[rotate=4*36+180] (p5) at (90+4*36:2.6) {\tiny $-$} ;
\node[rotate=4*36+180] (m5) at (90+4*36:2.3) {\tiny$+$} ;
\node[rotate=5*36+180] (p6) at (90+5*36:2.6) {\tiny $-$} ;
\node[rotate=5*36+180] (m6) at (90+5*36:2.3) {\tiny$-$} ;
\node[rotate=6*36+180] (p7) at (90+6*36:2.6) {\tiny $-$} ;
\node[rotate=6*36+180] (m7) at (90+6*36:2.3) {\tiny$+$} ;
\node[rotate=7*36+180] (p8) at (90+7*36:2.6) {\tiny $-$} ;
\node[rotate=7*36+180] (m8) at (90+7*36:2.3) {\tiny$+$} ;
\node[rotate=8*36] (p9) at (90+8*36:2.6) {\tiny $-$} ;
\node[rotate=8*36] (m9) at (90+8*36:2.3) {\tiny$+$} ;
\node[rotate=9*36] (p10) at (90+9*36:2.6) {\tiny $-$} ;
\node[rotate=9*36] (m10) at (90+9*36:2.3) {\tiny$+$} ;

\draw[dashed,->] (a1) to [bend right=30] (b1);
\draw[dashed,->] (a2) to [bend left=30] (c1);
\draw[dotted,->] (a1) to [bend right=30] (d1);
\draw[dotted,->] (a2) to [bend left=30] (e2);

\draw[dashed,->] (b1) to [bend right=30] (c1);
\draw[dashed,->] (b2) to [bend left=30] (e2);
\draw[dotted,->] (b1) to [bend right=30] (a1);
\draw[dotted,->] (b2) to [bend left=30] (d2);

\draw[dashed,->] (c1) to [bend right=30] (b1);
\draw[dashed,->] (c2) to [bend left=30] (d1);
\draw[dotted,->] (c1) to [bend left=30] (a2);
\draw[dotted,->] (c2) to [bend right=30] (e1);

\draw[dashed,->] (d1) to [bend right=30] (a1);
\draw[dashed,->] (d2) to [bend left=30] (b2);
\draw[dotted,->] (d1) to [bend left=30] (c2);
\draw[dotted,->] (d2) to [bend left=30] (e1);

\draw[dashed,->] (e1) to [bend right=30] (c2);
\draw[dashed,->] (e2) to [bend left=30] (a2);
\draw[dotted,->] (e1) to [bend left=30] (d2);
\draw[dotted,->] (e2) to [bend left=30] (b2);
\end{scope}

\end{tikzpicture}
\end{center}
\vspace{-0.2cm}
\caption{The house graph on the left. On the right, the \CSC{} $\C(D)$ of the (unique) chord diagram $D$ of the house graph (drawn in the middle). The dashed arc leaving an endpoint $x$ represents $+_{\C(D)}(x)$ and the dotted arc represents $-_{\C(D)}(x)$. For example, $-_{\C(D)}(a)=e$ and $+_{\C(D)}(a)=c$.\label{fig_CSC} 
}
\end{figure}
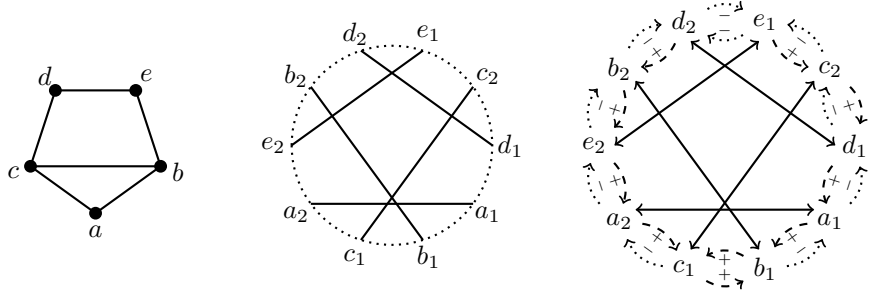

Let $\C(D)$ be a \CSC{} on the set $\chi$ of chords. Let $x_1$ and $x_2$ be the two endpoints of a chord $x$. We let $+_{\C(D)}(x_1,x_2)$ denote the factor of $\C(D)$ obtained by searching $\C(D)$ from $+_{\C(D)}(x_1)$ to $+_{\C(D)}(x_2)$. The factor $-_{\C(D)}(x_1,x_2)$ is defined similarly. In the example of~\autoref{fig_CSC}, we have $+_{\C(D)}(e_1,e_2)=c_2d_1a_1b_1c_1a_2$ and $-_{\C(D)}(e_1,e_2)=d_2b_2$.
We observe that $+_{\C(D)}(x_1,x_2)=+_{\C(D)}(x_2,x_1)^r$ and that either $+_{\C(D)}(x_1,x_2)=D(x_1,x_2)$ or $+_{\C(D)}(x_1,x_2)=D^r(x_1,x_2)$.

The following three lemmas state the complexity of basic operations on CSC's that the algorithm has to perform. Their proofs are straightforward from the data-structure description of CSC's. They have been formally discussed in \cite{GioanPTC14Circle} as Lemmas 5.5--5.7.

\begin{lemma}[Consecutivity test] \label{lem_consecutive_test}
Let $D$ be a chord diagram of a circle graph $G=(V,E)$. Given a subset $S$ of vertices of $G$ and $\C(D)$ the \CSC{} representing $D$, we can test in $O(|S|)$-time if $S$ is consecutive in $D$.
\end{lemma}

\begin{lemma}[Consecutivity preserving join] \label{lem_consecutive_join}
Let $D$ and $D'$ be chord diagrams of the circle graphs $G=(V,E)$ and $G'=(V',E')$, respectively, and let $S$ and $S'$ be consecutive chords in $D$ and $D'$, respectively. Let further $x$ and $y$ denote the bookends of $S$, and $x'$ and $y'$ the bookends of $S'$, and $\C(D)$ and $\C(D')$ the \CSC{}s representing $D$ and $D'$.  Then a \CSC{} representing a chord diagram of $(G,x)\otimes (G',x')$ such that $(S\setminus\{x\})\cup(S\setminus\{x'\})$ is consecutive and has bookends $y$ and $y'$ can be computed in $O(1)$ time.
\end{lemma}

\begin{lemma}[Chord insertion] \label{lem_chord_insertion}
Let $D$ be a chord diagram of a circle graph $G=(V,E)$ and let~$S$ be a subset of vertices of $G$ that is consecutive in $D$. Given $\C(D)$ the \CSC{} representing $D$ and the bookends $b$ and $b'$ of the factor of $D$ certifying that $S$ is consecutive, a \CSC{} of $G+x$ where the neighborhood of $x$ is $S$ can be computed in $O(1)$ time.
\end{lemma}

\subsection{PC-trees}

The data structure we use to encode the split-tree $\ST(G)$ of a circle graph $G$ is based on PC-trees~\cite{HsuMcC03PCTrees,FinkPR23Experimental}. We call it the \emph{split PC-tree} of $G$ and it is denoted $\PC(G)$
(see \autoref{fig_PCtree} for an example). It stores all marker vertices of $\ST(G)$ and encodes the structure of $\ST(G)$ on top of them as follows.
We select an arbitrary leaf of $\ST(G)$ as the root of $\PC(G)$, which we denoteby  $r$. Let $u$ be a node of the split-tree $\ST(G)$. For each marker vertex $q\in V(u)$, we store an outgoing arc to its opposite $q'$, which is a marker vertex of $G(u')$ for some node $u'$ adjacent to $u$ in $\ST(G)$.
Observe that this way, each tree-edge $e$ of $\ST(G)$ corresponds to a pair of oppositely oriented arcs. We assume moreover that for an arc $a$ of $\PC(G)$, the opposite arc, denoted $\bar{a}$, can be accessed in $O(1)$-time. The way marker vertices of $V(u)$ are stored depends on the type of the node $u$ in $\ST(G)$.
\begin{itemize}
\item Suppose $u$ is degenerate. Then $u$ is represented in $\PC(G)$ by a node object $\node(u)$ that stores (i)  the type of~$u$, (ii) a (doubly-linked) list $\LL(u)$ of the marker vertices in~$V(u)$  (iii) a pointer to the node of $V(u)$ whose arc points towards the root~$r$, and (iv) in case~$u$ is a star, a pointer to the marker vertex that is the center of the star.  Further, each marker vertex of~$V(u)$ is equipped with a pointer to the node object $\node(u)$. 
\item Suppose $u$ is not degenerate. Then $\PC(G)$ does not store any node object for $u$. Instead, the vertices of $V(u)$ are stored in a \CSC{} $\C(u)$ representing a chord diagram of $G(u)$.
For each vertex $x\in V(u)$, one of the two endpoints of the corresponding chord is made incident to the symmetric pointers representing the tree-edge $e$ of which $x$ is the extremity.
Further, a flag is used to mark the vertex whose arc points towards the root~$r$ of $\ST(G)$.
\end{itemize}

\begin{figure}[h]
\begin{center}
\bigskip
\begin{tikzpicture}[thick,scale=0.6]
\tikzstyle{sommet}=[circle, draw, fill=black, inner sep=0pt, minimum width=4pt]
\tikzstyle{newsommet}=[circle, draw, fill=green!50!black, inner sep=0pt, minimum width=4pt]

\begin{scope}[xshift=-6cm,yshift=3cm]

\node[] (5) at (90+360/5:1.5) {} ;
\draw[] (5) node[sommet]{};
\node[] (55) at (90+360/5:1.9) {$5$};

\node[] (1) at (80+2*360/5:1.5) {} ;
\draw[] (1) node[sommet]{};
\node[] (11) at (80+2*360/5:1.9) {$1$};

\node[] (4) at (85+1.62*360/5:2.7) {} ;
\draw[] (4) node[sommet]{};
\node[above] (44) at (85+1.62*360/5:2.7) {$4$};

\node[] (2) at (65+3*360/5:1.5) {} ;
\draw[] (2) node[sommet]{};
\node[] (22) at (65+3*360/5:1.9) {$2$};

\node[] (3) at (95+3*360/5:1.5) {} ;
\draw[] (3) node[sommet]{};
\node[] (33) at (95+3*360/5:1.9) {$3$};

\node[] (6) at (65+4*360/5:1.5) {} ;
\draw[] (6) node[sommet]{};
\node[] (66) at (65+4*360/5:1.9) {$6$};

\node[] (7) at (90+4*360/5:1.5) {} ;
\draw[] (7) node[sommet]{};
\node[] (77) at (90+4*360/5:1.9) {$7$};

\node[] (8) at (115+4*360/5:1.5) {} ;
\draw[] (8) node[sommet]{};
\node[] (88) at (115+4*360/5:1.9) {$8$};

\node[] (9) at (100:1.5) {} ;
\draw[] (9) node[sommet]{};
\node[] (99) at (100:1.9) {$9$};

\draw (1.center) -- (2.center) ;
\draw (1.center) -- (3.center) ;
\draw (1.center) -- (4.center) ;
\draw (1.center) -- (5.center) ;
\draw (2.center) -- (6.center) ;
\draw (2.center) -- (7.center) ;
\draw (2.center) -- (8.center) ;
\draw (3.center) -- (6.center) ;
\draw (3.center) -- (7.center) ;
\draw (3.center) -- (8.center) ;
\draw (5.center) -- (6.center) ;
\draw (5.center) -- (7.center) ;
\draw (5.center) -- (9.center) ;
\draw (6.center) -- (7.center) ;
\draw (6.center) -- (8.center) ;
\draw (6.center) -- (9.center) ;
\draw (7.center) -- (8.center) ;
\draw (7.center) -- (9.center) ;
\draw (8.center) -- (9.center) ;

\end{scope}

\begin{scope}[xshift=-6cm,yshift=-4cm,rotate=160]
\draw[dotted] (0,0) circle [radius=2.5cm] ;

\foreach \i in {1,...,16}{
	\node (v\i) at (\i*20:2.5){} ;
	}

\node (11) at (-15*20:2.9){$1$};
\node (12) at (-2*20:2.9){$1$};
\draw (-15*20:2.5) -- (-2*20:2.5);

\node (21) at (-18*20:2.9){$2$};
\node (22) at (-7*20:2.9){$2$};
\draw (-17*20:2.5) -- (-8*20:2.5);

\node (31) at (-17*20:2.9){$3$};
\node (32) at (-8*20:2.9){$3$};
\draw (-18*20:2.5) -- (-7*20:2.5);

\node (41) at (-1*20:2.9){$4$};
\node (42) at (-3*20:2.9){$4$};
\draw (-1*20:2.5) -- (-3*20:2.5);

\node (51) at (-10*20:2.9){$5$};
\node (52) at (-16*20:2.9){$5$};
\draw (-10*20:2.5) -- (-16*20:2.5);

\node (61) at (-12*20:2.9){$6$};
\node (62) at (-5*20:2.9){$6$};
\draw (-12*20:2.5) -- (-5*20:2.5);

\node (81) at (-13*20:2.9){$8$};
\node (82) at (-6*20:2.9){$8$};
\draw (-13*20:2.5) -- (-6*20:2.5);

\node (71) at (-11*20:2.9){$7$};
\node (72) at (-4*20:2.9){$7$};
\draw (-11*20:2.5) -- (-4*20:2.5);

\node (91) at (-9*20:2.9){$9$};
\node (92) at (-14*20:2.9){$9$};
\draw (-9*20:2.5) -- (-14*20:2.5);

\end{scope}

\begin{scope}[xshift=5.5cm,yshift=0cm,rotate=180]

\draw[fill=black!5, color=black!5] (0,0) circle [radius=3cm];

\node[] (a1) at (90-18+2*36:2.5) {$a$} ;
\node[] (a2) at (90-18+9*36:2.5) {$a$} ;
\node[] (aa1) at (90-18+2*36:2.2) {};
\node[] (aa2) at (90-18+9*36:2.2) {};
\draw[<->] (aa1.center) -- (aa2.center) ;

\node[] (b1) at (90-18+36:2.5) {$b$} ;
\node[] (b2) at (90-18+7*36:2.5) {$b$} ;
\node[] (bb1) at (90-18+36:2.2) {};
\node[] (bb2) at (90-18+7*36:2.2) {};
\draw[<->] (bb1.center) -- (bb2.center) ;

\node[] (c1) at (90-18:2.5) {$c$} ;
\node[] (c2) at (90-18+4*36:2.5) {$c$} ;
\node[] (cc1) at (90-18:2.3) {};
\node[] (cc2) at (90-18+4*36:2.3) {};
\draw[<->] (cc1.center) -- (cc2.center) ;

\node[] (d1) at (90-18+3*36:2.5) {$d$} ;
\node[] (d2) at (90-18+6*36:2.5) {$d$} ;
\node[] (dd1) at (90-18+3*36:2.2) {};
\node[] (dd2) at (90-18+6*36:2.2) {};
\draw[<->] (dd1.center) -- (dd2.center) ;

\node[red] (e1) at (90-18+5*36:2.5) {$e$} ;
\node[] (e2) at (90-18+8*36:2.5) {$e$} ;
\node[] (ee1) at (90-18+5*36:2.2) {};
\node[] (ee2) at (90-18+8*36:2.2) {};
\draw[<->] (ee1.center) -- (ee2.center) ;

\draw[dashed,->] (a1) to [bend right=30] (b1);
\draw[dashed,->] (a2) to [bend left=30] (c1);
\draw[dotted,->] (a1) to [bend right=30] (d1);
\draw[dotted,->] (a2) to [bend left=30] (e2);

\draw[dashed,->] (b1) to [bend right=30] (c1);
\draw[dashed,->] (b2) to [bend left=30] (e2);
\draw[dotted,->] (b1) to [bend right=30] (a1);
\draw[dotted,->] (b2) to [bend left=30] (d2);

\draw[dashed,->] (c1) to [bend right=30] (b1);
\draw[dashed,->] (c2) to [bend left=30] (d1);
\draw[dotted,->] (c1) to [bend left=30] (a2);
\draw[dotted,->] (c2) to [bend right=30] (e1);

\draw[dashed,->] (d1) to [bend right=30] (a1);
\draw[dashed,->] (d2) to [bend left=30] (b2);
\draw[dotted,->] (d1) to [bend left=30] (c2);
\draw[dotted,->] (d2) to [bend left=30] (e1);

\draw[dashed,->] (e1) to [bend right=30] (c2);
\draw[dashed,->] (e2) to [bend left=30] (a2);
\draw[dotted,->] (e1) to [bend left=30] (d2);
\draw[dotted,->] (e2) to [bend left=30] (b2);

\node[] (9) at (72+2*36:5) {$9$} ;
\draw[very thick,dotted,red,->] (a1) to [bend right=20] (9);
\draw[very thick,red,->] (9) to [bend right=20] (a1);

\node[] (5) at (90-18+36:5) {$5$} ;
\draw[very thick,dotted,red,->] (b1) to [bend right=20] (5);
\draw[very thick,red,->] (5) to [bend right=20] (b1);

\begin{scope}[shift=(72:5),rotate=-18]
\draw[fill=black!5, color=black!5] (-2.5,-0.5) rectangle (4.5,0.5);

\node[red] (oc) at (0,0) {$c'$};
\draw[very thick,dotted,red,->] (c1) to [bend right=20] (oc);
\draw[very thick,red,->] (oc) to [bend right=20] (c1);

\begin{scope}[]
\node[red] (o6) at (180:2) {$6'$};
\node[] (o7) at (0:2) {$7'$};
\node[] (o8) at (0:4) {$8'$};

\draw[dashed,->] (o6) to [bend left=15] (oc);
\draw[dashed,->] (oc) to [bend left=15] (o6);
\draw[dashed,->] (o7) to [bend left=15] (oc);
\draw[dashed,->] (oc) to [bend left=15] (o7);
\draw[dashed,->] (o8) to [bend left=15] (o7);
\draw[dashed,->] (o7) to [bend left=15] (o8);

\begin{scope}[shift=(180:2)]
\node[] (6) at (90:2.5) {$6$};
\draw[very thick,dotted,red,->] (o6) to [bend right=20] (6);
\draw[very thick,red,->] (6) to [bend right=20] (o6);
\end{scope}

\begin{scope}[shift=(0:2)]
\node[] (7) at (90:2.5) {$7$};
\draw[very thick,dotted,red,->] (o7) to [bend right=20] (7);
\draw[very thick,red,->] (7) to [bend right=20] (o7);
\end{scope}

\begin{scope}[shift=(0:4)]
\node[] (8) at (90:2.5) {$8$};
\draw[very thick,dotted,red,->] (o8) to [bend right=20] (8);
\draw[very thick,red,->] (8) to [bend right=20] (o8);
\end{scope}

\draw[blue,dotted,->] (o7) to [bend right=60] (oc);
\draw[blue,dotted,->] (o6) to [bend left=60] (oc);
\draw[blue,dotted,->] (o8) to [bend left=45] (oc);

\end{scope}
\end{scope}

\begin{scope}[shift=(288:5)]
\draw[fill=black!5, color=black!5] (-0.5,-0.5) rectangle (4.5,0.5);
\node[red] (od) at (0,0) {$d$};
\node[blue,above] (star) at (0,-0.2) {$\star$};

\begin{scope}[rotate=0]
\node[] (o2) at (0:2) {$2'$};
\node[] (o3) at (0:4) {$3'$};

\draw[very thick,dotted,red,->] (d2) to [bend right=20] (od);
\draw[very thick,red,->] (od) to [bend right=20] (d2);

\draw[dashed,->] (o2) to [bend left=15] (od);
\draw[dashed,->] (od) to [bend left=15] (o2);
\draw[dashed,->] (o3) to [bend left=15] (o2);
\draw[dashed,->] (o2) to [bend left=15] (o3);

\begin{scope}[shift=(0:2)]
\node[] (2) at (90:2.5) {$2$};
\draw[very thick,dotted,red,->] (o2) to [bend right=20] (2);
\draw[very thick,red,->] (2) to [bend right=20] (o2);
\end{scope}
%
\begin{scope}[shift=(0:4)]
\node[] (3) at (90:2.5) {$3$};
\draw[very thick,dotted,red,->] (o3) to [bend right=20] (3);
\draw[very thick,red,->] (3) to [bend right=20] (o3);
\end{scope}

\draw[blue,dotted,->] (o2) to [bend left=30] (od);
\draw[blue,dotted,->] (o3) to [bend left=30] (od);

\end{scope}
\end{scope}

\begin{scope}[shift=(252:5)]
\draw[fill=black!5, color=black!5] (0.5,-0.5) rectangle (-4.5,0.5);

\node[] (oe) at (0,0) {$e'$};
\node[blue] (star) at (-2,0.5) {$\star$};

\draw[very thick,dotted,red,->] (oe) to [bend right=20] (e1);
\draw[very thick,red,->] (e1) to [bend right=20] (oe);

\begin{scope}[rotate=0]
\node[red] (o1) at (180:2) {$1'$};
\node[] (o4) at (180:4) {$4'$};

\draw[dashed,->] (o4) to [bend left=15] (o1);
\draw[dashed,->] (o1) to [bend left=15] (o4);
\draw[dashed,->] (o1) to [bend left=15] (oe);
\draw[dashed,->] (oe) to [bend left=15] (o1);

\begin{scope}[shift=(0:-2)]
\node[red] (1) at (270:2.5) {$1$};
\draw[very thick,red,->] (o1) to [bend right=20] (1);
\draw[very thick,dotted,red,->] (1) to [bend right=20] (o1);
\end{scope}
%
\begin{scope}[shift=(0:-4)]
\node[] (4) at (90:2.5) {$4$};
\draw[very thick,dotted,red,->] (o4) to [bend right=20] (4);
\draw[very thick,red,->] (4) to [bend right=20] (o4);
\end{scope}

\draw[blue,dotted,->] (oe) to [bend right=60] (o1);
\draw[blue,dotted,->] (o4) to [bend left=60] (o1);

\end{scope}
\end{scope}

\end{scope}

\end{tikzpicture}
\end{center}
\vspace{-0.2cm}
\caption{
A circle graph $G$, a chord diagram of $G$ and the split PC-tree $\PC(G)$ rooted at leaf~$1$. The red plain arcs correspond to the tree-edges from nodes to their parent. The split-tree $\ST(G)$ contains a unique prime node $u$ labelled by the house graph and represented by a \CSC{}. The \emph{root of a node $u$}  is the extremity of the arcs between $u$ and its parent. Roots are  marker vertices or the  endpoints colored red. Every marker vertex of a degenerate node has a pointer towards its root. In a star node $u$, the center is identified by a blue star (which can be distinct from the root of $u$).
\label{fig_PCtree} 
}
\end{figure}
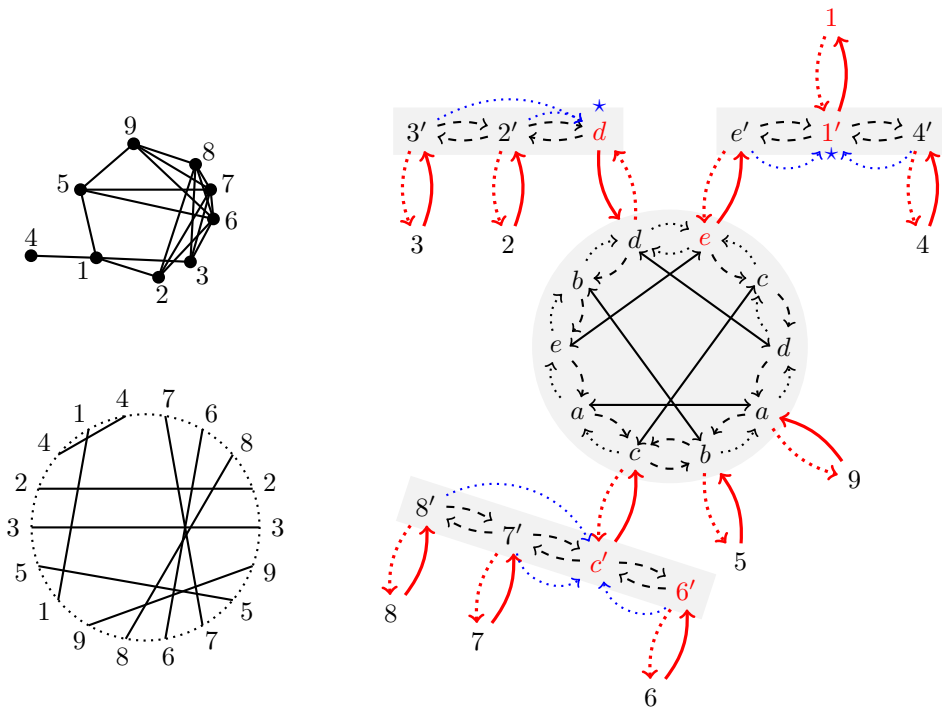

Observe that, given any marker vertex from~$V(u)$ of some degenerate node~$u$, we can determine the corresponding node object and thus also its parent arc in~$O(1)$ time.  On the other hand, due to the lack of a node object to represent a non-degenerate node $u$, finding the parent arc of $u$ is a costly operation. Indeed, given a marker vertex from~$V(u)$, one has to traverse~$\C(u)$ until  the endpoint is found whose flag indicates that it is incident to the parent arc. 
On the positive side, by~\autoref{lem_consecutive_join}, given two \CSC{}'s, it is possible to perform a node-join operation in $O(1)$ time and since the result is always a non-degenerate node, there is no need to update pointers to a node object, especially the pointer to the parent of the resulting node.

\section{Lexicographic Breadth-First-Search}
\label{sec:lexbfs}

The algorithm LexBFS (Lexicographic Breadth-First-Search) was introduced by Rose, Tarjan and Lueker~\cite{RoseTL76Algorithmic} to recognize chordal graphs in linear time. It has  since been used in many different contexts. We refer to~\cite{GioanPTC14Split} and references therein for a description of LexBFS. The circle graph recognition algorithm uses LexBFS as a pre-processing step to compute an ordering of the vertices of the input graph $G$. That ordering is then used to incrementally build $\PC(G)$ if $G$ is a circle graph by adding one vertex at a time.

A \emph{LexBFS ordering} of a graph $G$ is an ordering $\sigma$ in which LexBFS has visited the vertices of $G$. 
Let $y$ be a vertex of $G$. The \emph{slice} $S(y)$ is the largest factor of $\sigma$ starting at $y$ such that for every $x\prec_{\sigma} y$, $x$ is adjacent to $y$ if and only if $x$ is adjacent to every vertex of $S(y)$. It is known that 
the restriction $\sigma[S(y)]$ to the vertices of the slice $S(y)$ is a LexBFS ordering of $G[S(y)]$~\cite{CorneilOS09TheLBFS}. A vertex $x$ of a graph $G$ is \emph{good} if there exists a LexBFS ordering $\sigma$ whose last vertex is $x$. Due to the following lemma, good vertices are crucial to the LexBFS incremental recognition algorithm of circle graphs. This lemma is stated and proved in  \cite{GioanPTC14Circle} for the case of prime graphs. The case of arbitrary graph follows from the algorithm of  \cite{GioanPTC14Circle}. We provide a self-contained proof independent from their algorithm.

\begin{lemma}[Good vertex lemma]
\label{lem_good_vertex}
Let $G$ be a circle graph. If $x$ is a good vertex of $G$, then $G$ has a chord diagram $D$ in which $N(x)$ is consecutive.
\end{lemma}
\begin{proof}
Let $\sigma$ be a LexBFS ordering of $G$ ending at $x$. Let $z$ and $z'$ be the first two vertices of $\sigma$ in order. We prove the statement by induction on the number of slices containing $x$. Without loss of generality, we may assume that $G$ is connected. 

Suppose that for every vertex $w\notin\{z,z',x\}$, $S(w)$ is not a slice containing $x$.
For the sake of contradiction, assume that $N(x)$ is not consecutive in any chord diagram of $G=(V,E)$. Observe that in every chord diagram $D$, either one endpoint of $z$ appears in $D(x_1,x_2)$ and the other in $D(x_2,x_1)$, or the two endpoints of $z$ appear in one of $D(x_1,x_2)$ and $D(x_2,x_1)$. Without loss of generality, suppose that $D(x_1,x_2)$ contains at most one endpoint of $z$. Since $N(x)$ is not consecutive, at least one chord has its two endpoints in $D(x_1,x_2)$. Amongst such chords, let $y$ be the one occurring the earliest in $\sigma$. Observe  that by the choice of $y$, every neighbor $v$ of $y$ such that $v<_{\sigma} y$ is adjacent to $x$. Since $y<_{\sigma} x$, we also have that every non-neighbor of $v'$ of $y$ is a non-neighbor of $x$ and thereby $x\in S(y)$. Observe moreover that, by assumption on  $D(x_1,x_2)$, we have $y\neq z$. It follows that if $y\neq z'$, we have a contradiction. So assume that $y=z'$. Observe that by connectivity of $G$, we have that $z$ is universal in $G$ and thereby its chord has one endpoint in $D(x_1,x_2)$ and the other in $D(x_2,x_1)$. Then again, since $N(x)$ is not consecutive,  at least one chord has its two endpoints in $D(x_2,x_1)$. Amongst such chords, let $y'$ be the one occurring the earliest in $\sigma$. By the same argument than before we can prove that $S(y')$ is a slice containing $x$ and moreover $y'\notin\{z,z',x\}$, which is a contradiction.

Let us assume the property holds if the last vertex of a LexBFS ordering  belongs to $k\geq 2$ slices. 
Suppose that $x$ belongs to $k+1$ slices of $\sigma$ with $k\geq 2$. Let $y$ be the vertex of $G$ occurring last in $\sigma$ such that $y\notin\{z,z',x\}$ and $x\in S(y)$. If such a vertex does not exist, then by the discussion above we are done. Let us consider $B=\{v\in V\mid v\leq_{\sigma} y\}$ and $A=S(y)\cup \{y'\}$ where $y'$ is a neighbor of $y$ such that $y'<_{\sigma} y$. Observe that $\sigma[B]$ and $\sigma[A]$ are LexBFS orderings of $G[B]$ and $G[A]$ respectively ending at $y$ and $x$. Since $\sigma[B]$ has fewer slices containing $y$ than $\sigma$ has containing $x$, by the induction hypothesis, $G[B]$ has a chord diagram $D_B$ in which $N(y)\cap B$ is consecutive. Also the only slices containing $x$ in $\sigma[A]$ are $S(v)=A$ and $S(x)=\{x\}$. Again by induction hypothesis, $G[A]$ has a chord diagram $D_A$ in which $N(x)\cap A$ is consecutive. Hence, since $y'$ is universal in $G_A$, we may assume that $D_A(x_1,x_2)=X\cdot y'_2\cdot X'$ certifies the consecutiveness of $N(x)\cap A$.
Observe then that $D_A(y'_1,y'_2)\odot D_B(y_1,y_2)=D_A(y'_1,y'_2)\cdot D_B(y_1,y_2)\cdot D_A(y'_2,y'_1)\cdot D_B(y_2,y_1)$ is a chord diagram of $G$ in which $N(x)=(N(x)\cap A)\cup (N(y\cap B)$ is consecutive.
\end{proof}

Let $\sigma$ be a LexBFS ordering of a graph $G$ and let $u$ be a node of a \GLT{} $(T,\F)$ such that $\Gr(T,\F)=G$. Then we can define from $\sigma$ a LexBFS ordering $\sigma_u$ of $G(u)$ in the following way~\cite{GioanPTC14Split}: for two marker vertices $q$ and $q'$ of $G(u)$, let us denote $x$ the earliest vertex of $A(q)$ in $\sigma$ and $x'$ the earliest vertex of $A(q')$ in $\sigma$, then  $q<_{\sigma_u} q'$ if and only if $x<_{\sigma} x'$. 

\begin{lemma}[\cite{GioanPTC14Split}] \label{lem_LBFS_hereditary}
Let $\sigma$ be a LexBFS ordering of a graph $G$ and let $u$ be a node of $\ST(G)$. Then $\sigma_u$ is a LexBFS ordering of $G(u)$.
\end{lemma}

To illustrate \autoref{lem_LBFS_hereditary}, consider the graph $G$ of \autoref{fig_PCtree}. The numbering of its vertices from $1$ to $9$ is a LexBFS ordering. Considering the unique prime node $u$ of $\ST(G)$, then $\sigma_u=\langle e,d,b,c,a\rangle$ is a LexBFS ordering of the house graph $G(u)$.

\autoref{lem_good_vertex}, \autoref{lem_LBFS_hereditary} together with \autoref{lem_consecutive_join} allows us to sketch a circle graph recognition algorithm. The idea is to iteratively process the vertices of an input graph $G$ according to a LexBFS ordering in order to build $\PC(G)$. By \autoref{lem_good_vertex}, at each step, if $G$ is a circle graph, the neighbourhood of the vertex $x$ to be inserted is consecutive in some chord diagram of the so-far processed subgraph. By \autoref{lem_LBFS_hereditary}, this property is valid for the label graph of every node the current split PC-tree. Since adding a vertex may kill some existing split, we may have to perform node-join operations to compute the updated split PC-tree. By \autoref{lem_consecutive_join}, this can be done while preserving the consecutiveness of the neighbourhood of $x$. The challenge is to identify the part of the current split PC-tree that has to be shrunk into a single node. As we will see, either the input graph~$G$ is a circle graph and this can be done efficiently, or we are able to conclude that $G$  is not a circle graph.

\section{The vertex insertion algorithm}
\label{sec_algo}

Throughout this section, 
we assume that~$G=(V,E)$ is a circle graph, that $\PC(G)$ is the split PC-tree of $G$ encoding its split-tree $\ST(G)=(T,\F)$\footnote{Given that we can transform $\ST(G)$ into $\PC(G)$ and vice versa, depending on the context, if more convenient we may abusively switch between these two objects.}. 
We consider a new vertex $x$ with neighborhood $S\subseteq V$ that is the last vertex of a LexBFS ordering of $G'=G+x$, that is, $x$ is good in the graph $G'$. 
We assume without loss of generality that~$G$ is connected and $S \ne \emptyset$.

We let $T(S)$ denote the minimal subtree of~$T$ that contains all vertices in~$S$.  Let~$q$ be a leaf of $T$ or a marker vertex (possibly an endpoint) of some node $u$ of~$T$.  Following the work of Gioan et al.~\cite{GioanPTC14Split,GioanPTC14Circle}, the state of $q$ (with respect to $S$) is \emph{perfect} if $S \cap L(q) = A(q)$; \emph{empty} if $S \cap L(q) = \emptyset$; and \emph{mixed} otherwise (see \autoref{fig_marked_PCtree}).  We let $P(u)$ denote the set of perfect marker vertices of $u$ and $MP(u)$ the set of perfect or mixed marker vertices of $u$. For a degenerate node $u$, we define:
\[P^*(u)=\{q\in V(u)\mid q \mbox{ is perfect and not the center of a star}\},\]
\[E^*(u)=\{q\in V(u)\mid q \mbox{ is empty, or $q$ is perfect and the center of a star}\}.\]

\begin{figure}[ht]
\begin{center}
\bigskip
\begin{tikzpicture}[thick,scale=0.6]
\tikzstyle{sommet}=[circle, draw, fill=black, inner sep=0pt, minimum width=4pt]
\tikzstyle{newsommet}=[circle, draw, fill=green!50!black, inner sep=0pt, minimum width=4pt]

\begin{scope}[xshift=-6cm,yshift=3cm]

\node[] (5) at (90+360/5:1.5) {} ;
\draw[] (5) node[sommet]{};
\node[] (55) at (90+360/5:1.9) {$5$};

\node[] (1) at (80+2*360/5:1.5) {} ;
\draw[] (1) node[sommet]{};
\node[] (11) at (80+2*360/5:1.9) {$1$};

\node[] (4) at (85+1.62*360/5:2.7) {} ;
\draw[] (4) node[sommet]{};
\node[above] (44) at (85+1.62*360/5:2.7) {$4$};

\node[] (2) at (65+3*360/5:1.5) {} ;
\draw[] (2) node[sommet]{};
\node[] (22) at (65+3*360/5:1.9) {$2$};

\node[] (3) at (95+3*360/5:1.5) {} ;
\draw[] (3) node[sommet]{};
\node[] (33) at (95+3*360/5:1.9) {$3$};

\node[] (6) at (65+4*360/5:1.5) {} ;
\draw[] (6) node[sommet]{};
\node[] (66) at (65+4*360/5:1.9) {$6$};

\node[] (7) at (90+4*360/5:1.5) {} ;
\draw[] (7) node[sommet]{};
\node[] (77) at (90+4*360/5:1.9) {$7$};

\node[] (8) at (115+4*360/5:1.5) {} ;
\draw[] (8) node[sommet]{};
\node[] (88) at (115+4*360/5:1.9) {$8$};

\node[] (9) at (100:1.5) {} ;
\draw[] (9) node[sommet]{};
\node[] (99) at (100:1.9) {$9$};

\node[] (10) at (115+4*360/5:3) {} ;
\draw[green!50!black] (10) node[newsommet]{};
\node[green!50!black] (100) at (115+4*360/5:3.4) {$10$};
\draw[green!50!black] (10.center) -- (9.center);
\draw[green!50!black] (10.center) -- (7.center);
\draw[green!50!black] (10.center) to [bend left=60] (3.center);
\draw[green!50!black] (10.center) to [bend right=90] (5.center);

\draw (1.center) -- (2.center) ;
\draw (1.center) -- (3.center) ;
\draw (1.center) -- (4.center) ;
\draw (1.center) -- (5.center) ;
\draw (2.center) -- (6.center) ;
\draw (2.center) -- (7.center) ;
\draw (2.center) -- (8.center) ;
\draw (3.center) -- (6.center) ;
\draw (3.center) -- (7.center) ;
\draw (3.center) -- (8.center) ;
\draw (5.center) -- (6.center) ;
\draw (5.center) -- (7.center) ;
\draw (5.center) -- (9.center) ;
\draw (6.center) -- (7.center) ;
\draw (6.center) -- (8.center) ;
\draw (6.center) -- (9.center) ;
\draw (7.center) -- (8.center) ;
\draw (7.center) -- (9.center) ;
\draw (8.center) -- (9.center) ;

\end{scope}

\begin{scope}[xshift=-6cm,yshift=-4cm,rotate=160]
\draw[dotted] (0,0) circle [radius=2.5cm] ;

\foreach \i in {1,...,16}{
	\node (v\i) at (\i*18:2.5){} ;
	}

\node (11) at (-17*18:2.9){$1$};
\node (12) at (-2*18:2.9){$1$};
\draw (-17*18:2.5) -- (-2*18:2.5);

\node (21) at (-20*18:2.9){$2$};
\node (22) at (-7*18:2.9){$2$};
\draw (-20*18:2.5) -- (-7*18:2.5);

\node (31) at (-19*18:2.9){$3$};
\node (32) at (-9*18:2.9){$3$};
\draw (-19*18:2.5) -- (-9*18:2.5);

\node (41) at (-1*18:2.9){$4$};
\node (42) at (-3*18:2.9){$4$};
\draw (-1*18:2.5) -- (-3*18:2.5);

\node (51) at (-11*18:2.9){$5$};
\node (52) at (-18*18:2.9){$5$};
\draw (-11*18:2.5) -- (-18*18:2.5);

\node (61) at (-14*18:2.9){$6$};
\node (62) at (-5*18:2.9){$6$};
\draw (-14*18:2.5) -- (-5*18:2.5);

\node (71) at (-12*18:2.9){$7$};
\node (72) at (-4*18:2.9){$7$};
\draw (-12*18:2.5) -- (-4*18:2.5);

\node (81) at (-15*18:2.9){$8$};
\node (82) at (-6*18:2.9){$8$};
\draw (-15*18:2.5) -- (-6*18:2.5);

\node (91) at (-10*18:2.9){$9$};
\node (92) at (-16*18:2.9){$9$};
\draw (-10*18:2.5) -- (-16*18:2.5);

\node[green!50!black] (101) at (-13*18:2.9){$10$};
\node[green!50!black]  (102) at (-8*18:2.9){$10$};
\draw[green!50!black,very thick]  (-8*18:2.5) -- (-13*18:2.5);

\end{scope}

\begin{scope}[xshift=6cm,yshift=0cm,rotate=180]

\draw[fill=black!5, color=black!5] (0,0) circle [radius=3cm];

\node[green!50!black] (la1) at (90-18+2*36:3.3) {\tiny $\mathsf{P}$};
\node[green!50!black] (la2) at (90-18+2*36:4.2) {\tiny $\mathsf{M}$};
\node[] (a1) at (90-18+2*36:2.5) {$a$} ;
\node[] (a2) at (90-18+9*36:2.5) {$a$} ;
\node[] (aa1) at (90-18+2*36:2.2) {};
\node[] (aa2) at (90-18+9*36:2.2) {};
\draw[<->] (aa1.center) -- (aa2.center) ;

\node[green!50!black] (lb1) at (90-18+36:3.3) {\tiny $\mathsf{P}$};
\node[green!50!black] (lb2) at (90-18+36:4.2) {\tiny $\mathsf{M}$};
\node[] (b1) at (90-18+36:2.5) {$b$} ;
\node[] (b2) at (90-18+7*36:2.5) {$b$} ;
\node[] (bb1) at (90-18+36:2.2) {};
\node[] (bb2) at (90-18+7*36:2.2) {};
\draw[<->] (bb1.center) -- (bb2.center) ;

\node[green!50!black] (lc1) at (90-18:3.3) {\tiny $\mathsf{M}$};
\node[green!50!black] (lc2) at (90-18:4.2) {\tiny $\mathsf{M}$};
\node[] (c1) at (90-18:2.5) {$c$} ;
\node[] (c2) at (90-18+4*36:2.5) {$c$} ;
\node[] (cc1) at (90-18:2.3) {};
\node[] (cc2) at (90-18+4*36:2.3) {};
\draw[<->] (cc1.center) -- (cc2.center) ;

\node[green!50!black] (ld1) at (90-18+6*36:3.3) {\tiny $\mathsf{M}$};
\node[green!50!black] (ld2) at (90-18+6*36:4.2) {\tiny $\mathsf{M}$};
\node[] (d1) at (90-18+3*36:2.5) {$d$} ;
\node[] (d2) at (90-18+6*36:2.5) {$d$} ;
\node[] (dd1) at (90-18+3*36:2.2) {};
\node[] (dd2) at (90-18+6*36:2.2) {};
\draw[<->] (dd1.center) -- (dd2.center) ;

\node[green!50!black] (le1) at (90-18+5*36:3.3) {\tiny $\mathsf{\emptyset}$};
\node[green!50!black] (le2) at (90-18+5*36:4.2) {\tiny $\mathsf{M}$};
\node[red] (e1) at (90-18+5*36:2.5) {$e$} ;
\node[] (e2) at (90-18+8*36:2.5) {$e$} ;
\node[] (ee1) at (90-18+5*36:2.2) {};
\node[] (ee2) at (90-18+8*36:2.2) {};
\draw[<->] (ee1.center) -- (ee2.center) ;

\draw[dashed,->] (a1) to [bend right=30] (b1);
\draw[dashed,->] (a2) to [bend left=30] (c1);
\draw[dotted,->] (a1) to [bend right=30] (d1);
\draw[dotted,->] (a2) to [bend left=30] (e2);

\draw[dashed,->] (b1) to [bend right=30] (c1);
\draw[dashed,->] (b2) to [bend left=30] (e2);
\draw[dotted,->] (b1) to [bend right=30] (a1);
\draw[dotted,->] (b2) to [bend left=30] (d2);

\draw[dashed,->] (c1) to [bend right=30] (b1);
\draw[dashed,->] (c2) to [bend left=30] (d1);
\draw[dotted,->] (c1) to [bend left=30] (a2);
\draw[dotted,->] (c2) to [bend right=30] (e1);

\draw[dashed,->] (d1) to [bend right=30] (a1);
\draw[dashed,->] (d2) to [bend left=30] (b2);
\draw[dotted,->] (d1) to [bend left=30] (c2);
\draw[dotted,->] (d2) to [bend left=30] (e1);

\draw[dashed,->] (e1) to [bend right=30] (c2);
\draw[dashed,->] (e2) to [bend left=30] (a2);
\draw[dotted,->] (e1) to [bend left=30] (d2);
\draw[dotted,->] (e2) to [bend left=30] (b2);

\node[green!50!black]  (9) at (72+2*36:5) {$9$} ;
\draw[very thick,dotted,red,->] (a1) to [bend right=20] (9);
\draw[very thick,red,->] (9) to [bend right=20] (a1);

\node[green!50!black] (5) at (90-18+36:5) {$5$} ;
\draw[very thick,dotted,red,->] (b1) to [bend right=20] (5);
\draw[very thick,red,->] (5) to [bend right=20] (b1);

\begin{scope}[shift=(72:5),rotate=-18]
\draw[fill=black!5, color=black!5] (-2.5,-0.5) rectangle (4.5,0.5);

\node[red] (oc) at (0,0) {$c'$};
\draw[very thick,dotted,green!50!black,->] (c1) to [bend right=20] (oc);
\draw[very thick,green!50!black,->] (oc) to [bend right=20] (c1);

\begin{scope}[]
\node[] (o6) at (180:2) {$6'$};
\node[] (o7) at (0:2) {$7'$};
\node[] (o8) at (0:4) {$8'$};

\draw[dashed,->] (o6) to [bend left=15] (oc);
\draw[dashed,->] (oc) to [bend left=15] (o6);
\draw[dashed,->] (o7) to [bend left=15] (oc);
\draw[dashed,->] (oc) to [bend left=15] (o7);
\draw[dashed,->] (o8) to [bend left=15] (o7);
\draw[dashed,->] (o7) to [bend left=15] (o8);

\begin{scope}[shift=(180:2)]
\node[] (6) at (90:2.5) {$6$};
\node[green!50!black] (l61) at (90:1.7) {\tiny $\mathsf{M}$};
\node[green!50!black] (l62) at (90:0.8) {\tiny $\mathsf{\emptyset}$};
\draw[very thick,dotted,red,->] (o6) to [bend right=20] (6);
\draw[very thick,red,->] (6) to [bend right=20] (o6);
\end{scope}

\begin{scope}[shift=(0:2)]
\node[green!50!black]  (7) at (90:2.5) {$7$};
\node[green!50!black] (l71) at (90:1.7) {\tiny $\mathsf{M}$};
\node[green!50!black] (l72) at (90:0.8) {\tiny $\mathsf{P}$};
\draw[very thick,dotted,red,->] (o7) to [bend right=20] (7);
\draw[very thick,red,->] (7) to [bend right=20] (o7);
\end{scope}

\begin{scope}[shift=(0:4)]
\node[] (8) at (90:2.5) {$8$};
\node[green!50!black] (l81) at (90:1.7) {\tiny $\mathsf{M}$};
\node[green!50!black] (l82) at (90:0.8) {\tiny $\mathsf{\emptyset}$};
\draw[very thick,dotted,red,->] (o8) to [bend right=20] (8);
\draw[very thick,red,->] (8) to [bend right=20] (o8);
\end{scope}

\draw[blue,dotted,->] (o7) to [bend right=60] (oc);
\draw[blue,dotted,->] (o6) to [bend left=60] (oc);
\draw[blue,dotted,->] (o8) to [bend left=45] (oc);

\end{scope}
\end{scope}

\begin{scope}[shift=(288:5)]
\draw[fill=black!5, color=black!5] (-0.5,-0.5) rectangle (4.5,0.5);

\node[red] (od) at (0,0) {$d$};
\node[blue,above] (star) at (0,-0.2) {$\star$};

\begin{scope}[rotate=0]
\node[] (o2) at (0:2) {$2'$};
\node[] (o3) at (0:4) {$3'$};

\draw[very thick,dotted,green!50!black,->] (d2) to [bend right=20] (od);
\draw[very thick,green!50!black,->] (od) to [bend right=20] (d2);

\draw[dashed,->] (o2) to [bend left=15] (od);
\draw[dashed,->] (od) to [bend left=15] (o2);
\draw[dashed,->] (o3) to [bend left=15] (o2);
\draw[dashed,->] (o2) to [bend left=15] (o3);

\begin{scope}[shift=(0:2)]
\node[] (2) at (90:2.5) {$2$};
\node[green!50!black] (l21) at (90:1.7) {\tiny $\mathsf{M}$};
\node[green!50!black] (l22) at (90:0.8) {\tiny $\mathsf{\emptyset}$};
\draw[very thick,dotted,red,->] (o2) to [bend right=20] (2);
\draw[very thick,red,->] (2) to [bend right=20] (o2);
\end{scope}
%
\begin{scope}[shift=(0:4)]
\node[green!50!black]  (3) at (90:2.5) {$3$};
\node[green!50!black] (l31) at (90:1.7) {\tiny $\mathsf{M}$};
\node[green!50!black] (l32) at (90:0.8) {\tiny $\mathsf{P}$};
\draw[very thick,dotted,red,->] (o3) to [bend right=20] (3);
\draw[very thick,red,->] (3) to [bend right=20] (o3);
\end{scope}

\draw[blue,dotted,->] (o2) to [bend left=30] (od);
\draw[blue,dotted,->] (o3) to [bend left=30] (od);

\end{scope}
\end{scope}

\begin{scope}[shift=(252:5)]
\draw[fill=black!5, color=black!5] (0.5,-0.5) rectangle (-4.5,0.5);
\node[] (oe) at (0,0) {$e'$};
\node[blue] (star) at (-2,0.5) {$\star$};

\draw[very thick,dotted,red,->] (oe) to [bend right=20] (e1);
\draw[very thick,red,->] (e1) to [bend right=20] (oe);

\begin{scope}[rotate=0]
\node[red] (o1) at (180:2) {$1'$};
\node[] (o4) at (180:4) {$4'$};

\draw[dashed,->] (o4) to [bend left=15] (o1);
\draw[dashed,->] (o1) to [bend left=15] (o4);
\draw[dashed,->] (o1) to [bend left=15] (oe);
\draw[dashed,->] (oe) to [bend left=15] (o1);

\begin{scope}[shift=(0:-2)]
\node[red] (1) at (270:2.5) {$1$};
\node[green!50!black] (l11) at (270:1.7) {\tiny $\mathsf{M}$};
\node[green!50!black] (l12) at (270:0.8) {\tiny $\mathsf{\emptyset}$};
\draw[very thick,red,->] (o1) to [bend right=20] (1);
\draw[very thick,dotted,red,->] (1) to [bend right=20] (o1);
\end{scope}
%
\begin{scope}[shift=(0:-4)]
\node[] (4) at (90:2.5) {$4$};
\node[green!50!black] (l41) at (90:1.7) {\tiny $\mathsf{M}$};
\node[green!50!black] (l42) at (90:0.8) {\tiny $\mathsf{\emptyset}$};
\draw[very thick,dotted,red,->] (o4) to [bend right=20] (4);
\draw[very thick,red,->] (4) to [bend right=20] (o4);
\end{scope}

\draw[blue,dotted,->] (oe) to [bend right=60] (o1);
\draw[blue,dotted,->] (o4) to [bend left=60] (o1);

\end{scope}
\end{scope}

\end{scope}

\end{tikzpicture}
\end{center}
\caption{A circle graph $G'=G+10$ with vertex $10$ being a good vertex of $G'$ adjacent to $S=\{3,5,7,9\}$, a chord diagram of $G'$ and the split PC-tree $\PC(G)$ marked with respect to $\{3,5,7,9\}$. Every marker vertex is assigned green tags: $\mathsf{P}$ if it is perfect; $\mathsf{M}$ if it is mixed; $\emptyset$ otherwise. The two green tree-edges are mixed and $\PC(G)$ does not contain any hybrid node. Thereby endpoint $e$ is marked empty since $L(e)=\{1,4\}$ does not intersect $S$; the endpoint $a$ is marked perfect since $L(e)=\{9\}=A(e)\subseteq S$; and the endpoint $c$ is marker mixed sinsce $L(c)=\{6,7,8\}=A(c)$ while $S\cap L(c)\neq A(c)$.
\label{fig_marked_PCtree} 
}
\end{figure}
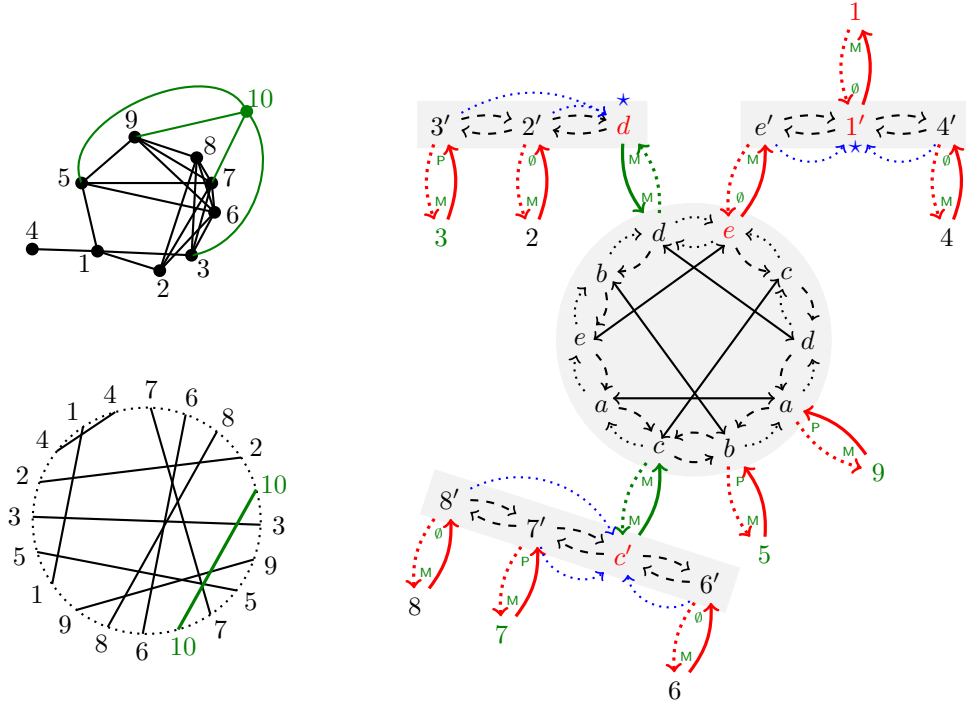

A node~$u$ of $T$ is \emph{hybrid} if every marker vertex~$q$ of $u$ has the property that $q$ is perfect or empty and its opposite is mixed.  A tree-edge is \emph{mixed} if both its extremities are mixed. It can be proved~\cite{GioanPTC14Split} that the subtree $T'$ of $T$ defined by the mixed tree-edges is unique and connected, it is called the \emph{fully-mixed subtree} of $\ST(G)$. The following lemma is central to the correctness of the circle graph recognition algorithm:

\begin{lemma}[\cite{GioanPTC14Circle}] \label{lem_fully_mixed}
Let $G'=G+x$ be a prime graph such that $x$ is a good vertex of $G'$ and $G$ is a circle graph. Then $G'$ is a circle graph if and only if for every node $u$ in $\ST(G)$, marked with respect to $N_{G'}(x)$, $G(u)$ has a chord diagram in which $MP(u)$ is consecutive, with the mixed marker vertices being bookends.
\end{lemma}

Let us discuss the above statement and provide some intuition of its correctness. First if $G'$ is a circle graph, then by \autoref{th_unique_prime}, it has a chord diagram~$D'$, which is unique up to reversal. Since $x$ is good, by \autoref{lem_good_vertex}, $N_{G'}(x)$ is consecutive in $D'$. Observe that $N_{G'}(x)$ is still consecutive in the chord diagram $D$ of $G$ obtained by removing the chord $x$ from $D'$. However, $G$ may not be a prime graph. If so we may retrieve the chord diagrams of every prime node of $\ST(G)$ by performing a series of \emph{circle-split} operations, the reverse operation of circle-join. When doing so, we can observe that the consecutive property is preserved as indicated in \cref{lem_fully_mixed}.
Conversely, assume that every node of $\ST(G)$ has a chord diagram satisfying the property stated in \cref{lem_fully_mixed}. Then by \cref{lem_consecutive_join_diagram}, iteratively performing for every tree-edge $uv$, a circle-join between the chord diagrams $D(u)$ and $D(v)$ with respect to the extremities of the tree-edge $uv$ eventually returns a chord diagram $D$ of $G$ in which $N_{G'}(x)$ is consecutive. It is then possible to insert in $D$ a chord $x$ crossing exactly $N_{G'}(x)$, certifying that $G'=G+x$ is a circle graph.

As announced earlier, the recognition algorithm iteratively inserts vertices according to a LexBFS ordering. Of course, for the sake of time complexity, we cannot afford at each insertion step contracting $\ST(G)$ in a single node. \cref{lem_fully_mixed} allows us to maintain a chord diagram for each prime node providing local certificates of membership to the class of circle graphs.

\subsection{Updating the split-tree}
\label{sub_split_tree_update}

Gioan \etal~\cite[Theorem 4.14, Proposition 4.15--4.20]{GioanPTC14Split} show that exactly one of the following cases occurs and in each case describe how to obtain the split-tree $\ST(G')$, with $G'=G+x$, from $\ST(G)=(T,\F)$.

\begin{enumerate}
\item \emph{$T$ contains a clique node $u$ whose marker vertices are all perfect; or $T$ contains a star node $u$ whose marker vertices are all empty except its center, which is perfect; or $T$ contains a unique hybrid node~$u$, which is prime.}

\smallskip
In each of these cases, the node~$u$ is unique and the split-tree $\ST(G')$ is obtained by adding to $T$ a leaf $\ell_x$ adjacent to $u$, adding to $G(u)$ a new marker vertex~$q_x$ adjacent to $P(u)$ and opposite to $\ell_x$.

We observe that the type of the updated node $u$ is the same as in $\ST(G)$. This implies that in the case $u$ is degenerate, by \autoref{theo_circle_split}, $G+x$ is a circle graph. So let us discuss the case that $u$ is a prime node. It is then represented in $\PC(G)$ by a CSC $\C(u)$ encoding its chord diagram, which is unique. In $\C(u)$, for $G+x$ to be a circle graph, $MP(u)$ has to be consecutive with the mixed marker vertices being the bookends (see~\autoref{lem_fully_mixed}). If so, adding $q_x$ consists in inserting a new chord whose extremities become the new bookends of $P(u)$.
Otherwise, we can conclude that $G+x$ is not a circle graph. 

\smallskip
\item \emph{$T$ contains a tree-edge~$e$ whose extremities are both perfect, and this edge is unique; or $T$ contains a tree-edge with one perfect and one empty extremity, and this edge is unique.}

\smallskip
In this case the tree-edge~$e$ is unique and $\ST(G')$ is obtained by: i) subdividing $e$ with a new node $u_x$; ii) adding a leaf $\ell_x$ adjacent to $u_x$; and iii) labelling $u_x$ by a clique of size $3$ (if both extremities of $e$ are perfect) or a star on two leaves whose center is opposite to the extremity of $e$ that is empty. Observe that since the new node $u_x$ is degenerate, by \autoref{theo_circle_split}, $G+x$ is a circle graph.

\smallskip
\item \emph{$T$ contains a hybrid node $u$, and this node is degenerate.}

\smallskip
In this case the node~$u$ is unique and $\ST(G')$ is obtained as follows. First perform a node-split on $u$ according to the split $(P^*(u),E^*(u))$, creating a new tree-edge~$e$, whose extremities are both perfect or one is perfect and the other is empty. Then $e$ can be processed as in the previous case. Again, since the modified nodes and the new node are degenerate, by \autoref{theo_circle_split}, $G+x$ is a circle graph.

\smallskip
\item \emph{$T$ contains a (unique) fully mixed subtree $M$.}

\smallskip
In this case, the fully mixed subtree is first cleaned and then contracted into a single node~$v$ by performing node-joins. When finally adding $x$, the graph $G(v)$ becomes a prime graph. 
Along the series of the node-joins, we need to compute chord diagrams that preserve the consecutive property of $N_{G'}(x)$ (see \autoref{lem_consecutive_join}). More precisely, the split-tree  $\ST(G')$ is obtained in three steps as follows: 
\smallskip
\begin{itemize}
\item[(i)] First, we clean the fully mixed subtree by performing on each degenerate node $u$ of $M$ a node-split with respect to the splits $(P^*(u),V(u)\setminus P^*(u))$ and/or $(E^*(u),V(u)\setminus E^*(u))$.

\smallskip
The resulting \GLT{} is denoted $\cl(\ST(G))$. Observe that the set of mixed tree-edges is left unchanged and thereby $M$ can still be abusively considered as the fully mixed subtree of $\cl(\ST(G))$. The difference with $\ST(G)$ is now that every degenerate node of $M$ contains at most one perfect and at most one empty marker (see \autoref{fig_cleaned_PCtree}).
Moreover by \autoref{lem_fully_mixed}, if $G+x$ is a circle graph, every node $u$ of $M$ has at most two mixed marker vertices which form the bookends of the factor certifying the consecutiveness of $MP(u)$. It follows that the degenerate nodes of $\cl(\ST(G))$ have bounded degree and each of them can be equipped with the accurate CSC in constant time. Otherwise, we can conclude that $G+x$ is not a circle graph.

\begin{figure}[h]
\begin{center}
\bigskip
\begin{tikzpicture}[thick,scale=0.65]
\tikzstyle{sommet}=[circle, draw, fill=black, inner sep=0pt, minimum width=4pt]
\tikzstyle{newsommet}=[circle, draw, fill=green!50!black, inner sep=0pt, minimum width=4pt]

\begin{scope}[xshift=-6cm,yshift=3cm]

\node[] (5) at (90+360/5:1.5) {} ;
\draw[] (5) node[sommet]{};
\node[] (55) at (90+360/5:1.9) {$5$};

\node[] (1) at (80+2*360/5:1.5) {} ;
\draw[] (1) node[sommet]{};
\node[] (11) at (80+2*360/5:1.9) {$1$};

\node[] (4) at (85+1.62*360/5:2.7) {} ;
\draw[] (4) node[sommet]{};
\node[above] (44) at (85+1.62*360/5:2.7) {$4$};

\node[] (2) at (65+3*360/5:1.5) {} ;
\draw[] (2) node[sommet]{};
\node[] (22) at (65+3*360/5:1.9) {$2$};

\node[] (3) at (95+3*360/5:1.5) {} ;
\draw[] (3) node[sommet]{};
\node[] (33) at (95+3*360/5:1.9) {$3$};

\node[] (6) at (65+4*360/5:1.5) {} ;
\draw[] (6) node[sommet]{};
\node[] (66) at (65+4*360/5:1.9) {$6$};

\node[] (7) at (90+4*360/5:1.5) {} ;
\draw[] (7) node[sommet]{};
\node[] (77) at (90+4*360/5:1.9) {$7$};

\node[] (8) at (115+4*360/5:1.5) {} ;
\draw[] (8) node[sommet]{};
\node[] (88) at (115+4*360/5:1.9) {$8$};

\node[] (9) at (100:1.5) {} ;
\draw[] (9) node[sommet]{};
\node[] (99) at (100:1.9) {$9$};

\node[] (10) at (115+4*360/5:3) {} ;
\draw[green!50!black] (10) node[newsommet]{};
\node[green!50!black] (100) at (115+4*360/5:3.4) {$10$};
\draw[green!50!black] (10.center) -- (9.center);
\draw[green!50!black] (10.center) -- (7.center);
\draw[green!50!black] (10.center) to [bend left=60] (3.center);
\draw[green!50!black] (10.center) to [bend right=90] (5.center);

\draw (1.center) -- (2.center) ;
\draw (1.center) -- (3.center) ;
\draw (1.center) -- (4.center) ;
\draw (1.center) -- (5.center) ;
\draw (2.center) -- (6.center) ;
\draw (2.center) -- (7.center) ;
\draw (2.center) -- (8.center) ;
\draw (3.center) -- (6.center) ;
\draw (3.center) -- (7.center) ;
\draw (3.center) -- (8.center) ;
\draw (5.center) -- (6.center) ;
\draw (5.center) -- (7.center) ;
\draw (5.center) -- (9.center) ;
\draw (6.center) -- (7.center) ;
\draw (6.center) -- (8.center) ;
\draw (6.center) -- (9.center) ;
\draw (7.center) -- (8.center) ;
\draw (7.center) -- (9.center) ;
\draw (8.center) -- (9.center) ;

\end{scope}

\begin{scope}[xshift=-6cm,yshift=-4cm,rotate=160]
\draw[dotted] (0,0) circle [radius=2.5cm] ;

\foreach \i in {1,...,16}{
	\node (v\i) at (\i*18:2.5){} ;
	}

\node (11) at (-17*18:2.9){$1$};
\node (12) at (-2*18:2.9){$1$};
\draw (-17*18:2.5) -- (-2*18:2.5);

\node (21) at (-20*18:2.9){$2$};
\node (22) at (-7*18:2.9){$2$};
\draw (-20*18:2.5) -- (-7*18:2.5);

\node (31) at (-19*18:2.9){$3$};
\node (32) at (-9*18:2.9){$3$};
\draw (-19*18:2.5) -- (-9*18:2.5);

\node (41) at (-1*18:2.9){$4$};
\node (42) at (-3*18:2.9){$4$};
\draw (-1*18:2.5) -- (-3*18:2.5);

\node (51) at (-11*18:2.9){$5$};
\node (52) at (-18*18:2.9){$5$};
\draw (-11*18:2.5) -- (-18*18:2.5);

\node (61) at (-14*18:2.9){$6$};
\node (62) at (-5*18:2.9){$6$};
\draw (-14*18:2.5) -- (-5*18:2.5);

\node (71) at (-12*18:2.9){$7$};
\node (72) at (-4*18:2.9){$7$};
\draw (-12*18:2.5) -- (-4*18:2.5);

\node (81) at (-15*18:2.9){$8$};
\node (82) at (-6*18:2.9){$8$};
\draw (-15*18:2.5) -- (-6*18:2.5);

\node (91) at (-10*18:2.9){$9$};
\node (92) at (-16*18:2.9){$9$};
\draw (-10*18:2.5) -- (-16*18:2.5);

\node[green!50!black] (101) at (-13*18:2.9){$10$};
\node[green!50!black]  (102) at (-8*18:2.9){$10$};
\draw[green!50!black,very thick]  (-8*18:2.5) -- (-13*18:2.5);

\end{scope}

\begin{scope}[xshift=6cm,yshift=0cm,rotate=180]

\draw[fill=black!5, color=black!5] (0,0) circle [radius=3cm];

\node[green!50!black] (la1) at (90-18+2*36:3.3) {\tiny $\mathsf{P}$};
\node[green!50!black] (la2) at (90-18+2*36:4.2) {\tiny $\mathsf{M}$};
\node[] (a1) at (90-18+2*36:2.5) {$a$} ;
\node[] (a2) at (90-18+9*36:2.5) {$a$} ;
\node[] (aa1) at (90-18+2*36:2.2) {};
\node[] (aa2) at (90-18+9*36:2.2) {};
\draw[<->] (aa1.center) -- (aa2.center) ;

\node[green!50!black] (lb1) at (90-18+36:3.3) {\tiny $\mathsf{P}$};
\node[green!50!black] (lb2) at (90-18+36:4.2) {\tiny $\mathsf{M}$};
\node[] (b1) at (90-18+36:2.5) {$b$} ;
\node[] (b2) at (90-18+7*36:2.5) {$b$} ;
\node[] (bb1) at (90-18+36:2.2) {};
\node[] (bb2) at (90-18+7*36:2.2) {};
\draw[<->] (bb1.center) -- (bb2.center) ;

\node[green!50!black] (lc1) at (90-18:3.3) {\tiny $\mathsf{M}$};
\node[green!50!black] (lc2) at (90-18:4.2) {\tiny $\mathsf{M}$};
\node[] (c1) at (90-18:2.5) {$c$} ;
\node[] (c2) at (90-18+4*36:2.5) {$c$} ;
\node[] (cc1) at (90-18:2.3) {};
\node[] (cc2) at (90-18+4*36:2.3) {};
\draw[<->] (cc1.center) -- (cc2.center) ;

\node[green!50!black] (ld1) at (90-18+6*36:3.3) {\tiny $\mathsf{M}$};
\node[green!50!black] (ld2) at (90-18+6*36:4.2) {\tiny $\mathsf{M}$};
\node[] (d1) at (90-18+3*36:2.5) {$d$} ;
\node[] (d2) at (90-18+6*36:2.5) {$d$} ;
\node[] (dd1) at (90-18+3*36:2.2) {};
\node[] (dd2) at (90-18+6*36:2.2) {};
\draw[<->] (dd1.center) -- (dd2.center) ;

\node[green!50!black] (le1) at (90-18+5*36:3.3) {\tiny $\mathsf{\emptyset}$};
\node[green!50!black] (le2) at (90-18+5*36:4.2) {\tiny $\mathsf{M}$};
\node[red] (e1) at (90-18+5*36:2.5) {$e$} ;
\node[] (e2) at (90-18+8*36:2.5) {$e$} ;
\node[] (ee1) at (90-18+5*36:2.2) {};
\node[] (ee2) at (90-18+8*36:2.2) {};
\draw[<->] (ee1.center) -- (ee2.center) ;

\draw[dashed,->] (a1) to [bend right=30] (b1);
\draw[dashed,->] (a2) to [bend left=30] (c1);
\draw[dotted,->] (a1) to [bend right=30] (d1);
\draw[dotted,->] (a2) to [bend left=30] (e2);

\draw[dashed,->] (b1) to [bend right=30] (c1);
\draw[dashed,->] (b2) to [bend left=30] (e2);
\draw[dotted,->] (b1) to [bend right=30] (a1);
\draw[dotted,->] (b2) to [bend left=30] (d2);

\draw[dashed,->] (c1) to [bend right=30] (b1);
\draw[dashed,->] (c2) to [bend left=30] (d1);
\draw[dotted,->] (c1) to [bend left=30] (a2);
\draw[dotted,->] (c2) to [bend right=30] (e1);

\draw[dashed,->] (d1) to [bend right=30] (a1);
\draw[dashed,->] (d2) to [bend left=30] (b2);
\draw[dotted,->] (d1) to [bend left=30] (c2);
\draw[dotted,->] (d2) to [bend left=30] (e1);

\draw[dashed,->] (e1) to [bend right=30] (c2);
\draw[dashed,->] (e2) to [bend left=30] (a2);
\draw[dotted,->] (e1) to [bend left=30] (d2);
\draw[dotted,->] (e2) to [bend left=30] (b2);

\node[green!50!black]  (9) at (72+2*36:5) {$9$} ;
\draw[very thick,dotted,red,->] (a1) to [bend right=20] (9);
\draw[very thick,red,->] (9) to [bend right=20] (a1);

\node[green!50!black] (5) at (90-18+36:5) {$5$} ;
\draw[very thick,dotted,red,->] (b1) to [bend right=20] (5);
\draw[very thick,red,->] (5) to [bend right=20] (b1);

\begin{scope}[shift=(72:5),rotate=-18]

\begin{scope}[xshift=-2cm,yshift=2.5cm]
\draw[fill=black!5, color=black!5] (-2.5,-0.5) rectangle (2.5,0.5);

\node[red] (oqq) at (0,0) {$q'$};

\begin{scope}[]
\node[] (o6) at (180:2) {$6'$};
\node[] (o8) at (0:2) {$8'$};

\draw[dashed,->] (o6) to [bend left=15] (oqq);
\draw[dashed,->] (oqq) to [bend left=15] (o6);
\draw[dashed,->] (o8) to [bend left=15] (oqq);
\draw[dashed,->] (oqq) to [bend left=15] (o8);

\begin{scope}[shift=(180:2)]
\node[] (6) at (90:2.5) {$6$};
\node[green!50!black] (l61) at (90:1.7) {\tiny $\mathsf{M}$};
\node[green!50!black] (l62) at (90:0.8) {\tiny $\mathsf{\emptyset}$};
\draw[very thick,dotted,red,->] (o6) to [bend right=20] (6);
\draw[very thick,red,->] (6) to [bend right=20] (o6);
\end{scope}

\begin{scope}[shift=(0:2)]
\node[] (8) at (90:2.5) {$8$};
\node[green!50!black] (l81) at (90:1.7) {\tiny $\mathsf{M}$};
\node[green!50!black] (l82) at (90:0.8) {\tiny $\mathsf{\emptyset}$};
\draw[very thick,dotted,red,->] (o8) to [bend right=20] (8);
\draw[very thick,red,->] (8) to [bend right=20] (o8);
\end{scope}

\draw[blue,dotted,->] (o6) to [bend left=60] (oqq);
\draw[blue,dotted,->] (o8) to [bend right=60] (oqq);

\node[green!50!black] (l61) at (270:1.7) {\tiny $\mathsf{\emptyset}$};
\node[green!50!black] (l62) at (270:0.8) {\tiny $\mathsf{M}$};

\end{scope}
\end{scope}

\draw[fill=black!5, color=black!5] (-2.5,-0.5) rectangle (2.5,0.5);

\node[red] (oc) at (0,0) {$c'$};
\draw[very thick,dotted,green!50!black,->] (c1) to [bend right=20] (oc);
\draw[very thick,green!50!black,->] (oc) to [bend right=20] (c1);

\begin{scope}[]
\node[] (oq) at (180:2) {$q$};
\node[] (o7) at (0:2) {$7'$};

\draw[dashed,->] (oq) to [bend left=15] (oc);
\draw[dashed,->] (oc) to [bend left=15] (oq);
\draw[dashed,->] (o7) to [bend left=15] (oc);
\draw[dashed,->] (oc) to [bend left=15] (o7);

\begin{scope}[shift=(0:2)]
\node[green!50!black]  (7) at (90:2.5) {$7$};
\node[green!50!black] (l71) at (90:1.7) {\tiny $\mathsf{M}$};
\node[green!50!black] (l72) at (90:0.8) {\tiny $\mathsf{P}$};
\draw[very thick,dotted,red,->] (o7) to [bend right=20] (7);
\draw[very thick,red,->] (7) to [bend right=20] (o7);
\end{scope}

\draw[blue,dotted,->] (o7) to [bend right=60] (oc);
\draw[blue,dotted,->] (oq) to [bend left=60] (oc);

\end{scope}

\draw[very thick,dotted,red,->] (oq) to [bend right=20] (oqq);
\draw[very thick,red,->] (oqq) to [bend right=20] (oq);

\end{scope}

\begin{scope}[shift=(288:5)]
\draw[fill=black!5, color=black!5] (-0.5,-0.5) rectangle (4.5,0.5);

\node[red] (od) at (0,0) {$d$};
\node[blue,above] (star) at (0,-0.2) {$\star$};

\begin{scope}[rotate=0]
\node[] (o2) at (0:2) {$2'$};
\node[] (o3) at (0:4) {$3'$};

\draw[very thick,dotted,green!50!black,->] (d2) to [bend right=20] (od);
\draw[very thick,green!50!black,->] (od) to [bend right=20] (d2);

\draw[dashed,->] (o2) to [bend left=15] (od);
\draw[dashed,->] (od) to [bend left=15] (o2);
\draw[dashed,->] (o3) to [bend left=15] (o2);
\draw[dashed,->] (o2) to [bend left=15] (o3);

\begin{scope}[shift=(0:2)]
\node[] (2) at (90:2.5) {$2$};
\node[green!50!black] (l21) at (90:1.7) {\tiny $\mathsf{M}$};
\node[green!50!black] (l22) at (90:0.8) {\tiny $\mathsf{\emptyset}$};
\draw[very thick,dotted,red,->] (o2) to [bend right=20] (2);
\draw[very thick,red,->] (2) to [bend right=20] (o2);
\end{scope}
%
\begin{scope}[shift=(0:4)]
\node[green!50!black]  (3) at (90:2.5) {$3$};
\node[green!50!black] (l31) at (90:1.7) {\tiny $\mathsf{M}$};
\node[green!50!black] (l32) at (90:0.8) {\tiny $\mathsf{P}$};
\draw[very thick,dotted,red,->] (o3) to [bend right=20] (3);
\draw[very thick,red,->] (3) to [bend right=20] (o3);
\end{scope}

\draw[blue,dotted,->] (o2) to [bend left=30] (od);
\draw[blue,dotted,->] (o3) to [bend left=30] (od);

\end{scope}
\end{scope}

\begin{scope}[shift=(252:5)]
\draw[fill=black!5, color=black!5] (0.5,-0.5) rectangle (-4.5,0.5);
\node[] (oe) at (0,0) {$e'$};
\node[blue] (star) at (-2,0.5) {$\star$};

\draw[very thick,dotted,red,->] (oe) to [bend right=20] (e1);
\draw[very thick,red,->] (e1) to [bend right=20] (oe);

\begin{scope}[rotate=0]
\node[red] (o1) at (180:2) {$1'$};
\node[] (o4) at (180:4) {$4'$};

\draw[dashed,->] (o4) to [bend left=15] (o1);
\draw[dashed,->] (o1) to [bend left=15] (o4);
\draw[dashed,->] (o1) to [bend left=15] (oe);
\draw[dashed,->] (oe) to [bend left=15] (o1);

\begin{scope}[shift=(0:-2)]
\node[red] (1) at (270:2.5) {$1$};
\node[green!50!black] (l11) at (270:1.7) {\tiny $\mathsf{M}$};
\node[green!50!black] (l12) at (270:0.8) {\tiny $\mathsf{\emptyset}$};
\draw[very thick,red,->] (o1) to [bend right=20] (1);
\draw[very thick,dotted,red,->] (1) to [bend right=20] (o1);
\end{scope}
%
\begin{scope}[shift=(0:-4)]
\node[] (4) at (90:2.5) {$4$};
\node[green!50!black] (l41) at (90:1.7) {\tiny $\mathsf{M}$};
\node[green!50!black] (l42) at (90:0.8) {\tiny $\mathsf{\emptyset}$};
\draw[very thick,dotted,red,->] (o4) to [bend right=20] (4);
\draw[very thick,red,->] (4) to [bend right=20] (o4);
\end{scope}

\draw[blue,dotted,->] (oe) to [bend right=60] (o1);
\draw[blue,dotted,->] (o4) to [bend left=60] (o1);

\end{scope}
\end{scope}

\end{scope}

\end{tikzpicture}
\end{center}
\caption{A circle graph $G'=G+10$ with vertex $10$ being a good vertex of $G'$ adjacent to $\{3,5,7,9\}$, a chord diagram of $G'$ and the cleaned split PC-tree $\cl(\PC(G))$ marked with respect to $\{3,5,7,9\}$ \label{fig_cleaned_PCtree} 
}
\end{figure}
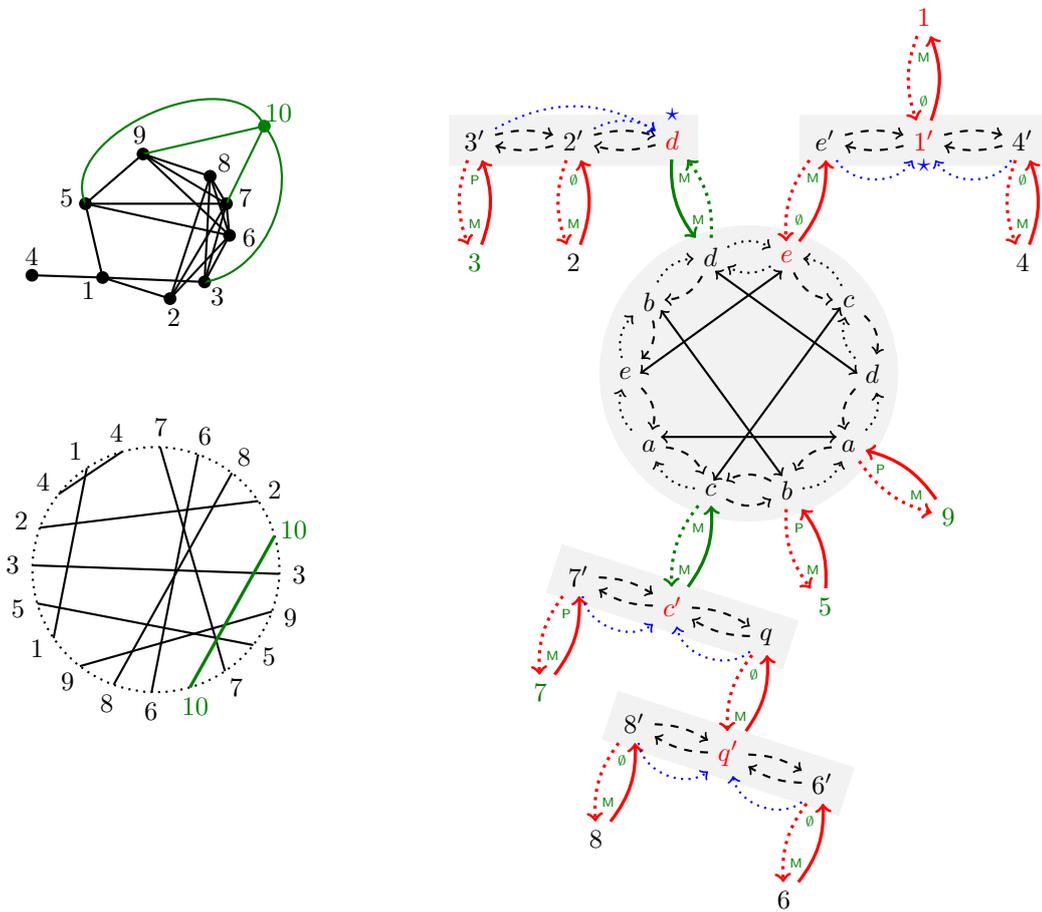

\smallskip
\item[(ii)] Then, we contract $M$ into a single node $v$ by performing for every mixed tree-edge $e=uu'$ in $M$ a node-join on $u$ and $u'$.
 
 \smallskip
As discussed above, each node-join has to return a chord diagram that preserves the consecutive property of $N_{G'}(x)$. Testing if such a chord diagram exists and computing it, can be done efficiently using the \CSC{} data-structure (see \autoref{lem_consecutive_join}). If there is no such chord diagram, then the algorithm stops and concludes that the graph $G+x$ is not a circle graph.

\smallskip
\item[(iii)] Finally, in the contracted \GLT{}, we attach a leaf $\ell_x$ adjacent to node $v$; and we add to $G(v)$ a new marker $q_x$ adjacent to the marker vertices of $P(u)$ and opposite to $\ell_x$.

\smallskip
In this last step, adding $q_x$ consists in inserting a new chord in the chord diagram of $G(v)$, which is possible since it satisfies the consecutive property of $N_{G'}(x)$ (see \autoref{lem_chord_insertion}). 
\end{itemize}
\end{enumerate}

\subsection{Recognition algorithm}
\label{sub_recognition}

Given that $\ST(G)$ (or $\PC(G)$) has been correctly marked with respect to $N_{G'}(x)$, where $G'=G+x$, the correctness of the updating algorithm described in \autoref{sub_split_tree_update} is proved by Gioan \emph{et al.}~\cite{GioanPTC14Circle} and has been discussed along its description. However, its complexity analysis relies on the split PC-tree data-structure. 
To achieve an overall linear running time, we implement the updating algorithm in such a way that each vertex $x$ is processed in amortized~$O(|N_{G'}(x)|)$ time.  Compared to the works of Gioan et al.~\cite{GioanPTC14Split,GioanPTC14Circle}, the main challenge is that, given a marker vertex~$q$ that belongs to a prime node~$u$, our split PC-tree data structure does not allow to efficiently determine the parent of~$u$.  Using the fact that such nodes are prime, we have a chord diagram for them, and~$x$ is inserted according to a LexBFS order allows us to anyway determine the parent without sacrificing the running time at least in all cases where this information is needed.  Using the information acquired in the first step, the second step can be implemented in much the same way as in the work of Gioan et al.~\cite{GioanPTC14Circle}.  The major difference is that, in contrast to the data structure they use, ours allows to implement the contraction of fully mixed edges in Case 4 in constant time as there is no need to propagate the parent pointer information.
Let us recall that, given a graph $G'=G+x$ such that $x$ is a good vertex and $G$ is a circle graph, and given the split PC-tree $\PC(G)$ encoding $ST(G)=(T,\F)$,
 the insertion algorithm works in two steps:
\smallskip
\begin{enumerate}
\item compute $T(N_{G'}(x))$, the minimal subtree of $T$ covering $N_{G'}(x)$;
\item identify the cases according to the description of \autoref{sub_split_tree_update} and update the split-tree correspondingly.
\end{enumerate}

We start discussing the implementation of the first step. For the sake of complexity efficiency, as stated by \autoref{lem_compute_subtree}, if $G'$ is a circle graph, we guarantee to compute $T(N_{G'}(x))$ in the expected time complexity. Otherwise, the algorithm may be able to conclude, already at that step, that 
$G'=G+x$ is not a circle graph. The reader has to keep in mind that if $\PC(G)$ encodes the split-tree $\ST(G)=(T,\mathcal{F})$, the tree $T$ is not explicitly stored in $\PC(G)$ since we are lacking node objects to represent the prime nodes of $\ST(G)$. However the tree-edges of $T$ are present in $\PC(G)$. The algorithm identifies these tree-edges.

\begin{lemma} \label{lem_compute_subtree}
Let $G'=G+x$ be a graph such that $x$ is a good vertex of $G'$ and $G$ is a circle graph.
Let $\PC(G)$ be the split PC-tree of $G$ encoding the split-tree $\ST(G)=(T,\mathcal{F})$. 
In time $O(|T(N_{G'}(x))|)$, we can either compute $T(N_{G'}(x))$ or conclude that $G'$ is not a circle graph.
\end{lemma}
\begin{proof}
We first describe the algorithm and then analyze its complexity. 
We let $\mathcal{L}(G)$ denote the leaves of $\PC(G)$, $\mathcal{D}(G)$ the set of node objects representing the degenerate nodes of $\ST(G)$ in $\PC(G)$, and $\mathcal{C}(G)$ the set of chord endpoints that are involved in the chord diagrams representing the prime nodes of $\ST(G)$.  We set $\mathcal{R}(G)=\mathcal{L}(G)\cup \mathcal{D}(G)\cup \mathcal{C}(G)$. To identify the tree-edges of $T(N_{G'}(x))$, we search $\PC(G)$ through the set $\mathcal{R}(G)$. The algorithm is the following:
\begin{enumerate}
\item Initially every leaf in $N_{G'}(x)$ is marked as \emph{active}. Every element of $\mathcal{R}(G)\setminus N_{G'}(x)$ is considered as \emph{inactive}. The set of active elements of $\mathcal{R}(G)$ is managed in a queue $\queue(G)$. In parallel, we maintain a set  of \emph{visited} elements of $\mathcal{R}(G)$, denoted $\visited(G)$. Recall that the root $r$ of $\PC(G)$ is a leaf.

\smallskip
\item We perform a first bottom-up traversal of $\PC(G)$ from the leaves in $\queue(G)$. The objective of this step is to identify a set $\mathfrak{E}$ of tree-edges containing those of $T(N_{G'}(x))$. 
As long as either $r\notin\visited(G)$ and $|\queue(G)|\geq 2$, or $r\in\visited(G)$ and $\queue(G)\neq\emptyset$, 
we pick the head element $a\in\queue(G)$, remove it from $\queue(G)$ and add it to $\visited(G)$.
If $a\neq r$, there are three cases to consider:
\begin{itemize}
\item $a\in \mathcal{L}(G)$: Let $q$ be the opposite of $a$. The tree-edge whose extremities are $a$ and $q$ is added to $\mathfrak{E}$. Observe that $q$ is either a marker vertex of a degenerate node or $q$ is the extremity of a chord of the chord diagram representing a prime node). In the former case, if
$u\notin\queue(G)\cup\visited(G)$, then add $u$ to $\queue(G)$. 
In the case that $q\in\C(G)$, if
$q\notin\queue(G)\cup\visited(G)$, then we add $q$ to $\queue(G)$. 

\item $a\in \mathcal{D}(G)$: Let $p$ be the root marker vertex of the degenerate node $a$ and $q$ be the opposite of $p$. We proceed with $q$ as in the previous case. %

\item $a\in \mathcal{C}(G)$: Let $\{a,a'\}$ be the corresponding chord. Let $u$ denote the node of $\ST(G)$ containing that chord and $\CD(u)$ be the corresponding CSC stored in $\PC(G)$.
For each chord endpoint $b$ that is consecutive to $a$ or $a'$ in $\CD(u)$, we check if $b$ is the flagged root of $\CD(u)$. If so, let $q$ be $b$'s opposite. We proceed with $q$ as in the previous cases. Moreover we add to $\mathfrak{E}$ a so-called \emph{fake tree-edge} whose extremities are $a$ and $b$ to later retrieve this flagged root information.
\end{itemize}

\smallskip
\item We perform a second bottom-up traversal of $\PC(G)$ from the leaves in $N_{G'}(x)$. The objective is to
link the identified tree-edges of $\mathfrak{E}$. Observe that during the previous search, when visiting an endpoint of a chord in a CSC $\CD(u)$, the flagged root $b$ of $\CD(u)$ may or may not have been identified. Suppose it was when visiting the endpoint $a$. Then a fake tree-edge between $a$ and $b$ was added to $\mathfrak{E}$. We now search from $a$ the set of consecutive endpoints that were visited during the previous search. For each of them, we add to $\mathfrak{E}$ a new fake tree-edge to $a$.

At this point, we have computed a forest whose edges are $\mathfrak{E}$.
 Observe that in this implementation, if the root endpoint $b$ of a CSC $\CD(u)$ is identified, then node $u$ of $\ST(G)$ is represented by a star whose root is $b$ and possibly a set of isolated endpoints. 
 
\smallskip
\item The last step consists in checking whether $\mathfrak{E}$ allows to retrieve $T(N_{G'}(x))$. We first check whether the resulting set $\mathfrak{E}$ of (fake) tree-edges form a connected tree. If not, we return that $G'$ is not a circle graph. Otherwise, let $T(\mathfrak{E})$ denote that tree. We may have to clear the path attached to the root $s$ of $T(\mathfrak{E})$. Suppose that $s$ is distinct from the root of $\PC(G)$ or its is the root $r$ of $\PC(G)$ (which is a leaf of $\PC(G)$) but the corresponding vertex does not belong to $N_{G'}(x)$. Then remove from $T(\mathfrak{E})$, the tree-edges of the path from that root to its first descendant with at least two children. This step can be completed by an extra search.
\end{enumerate}

Since the described algorithm basically consists in three traversals, its complexity is $O(|\mathfrak{E}|)$.  The key point concerning the running time is that the round-robin execution of the upward searches in step 2 guarantees that $|\mathfrak E| \le 2 \cdot |T(N_{G'}(x))|$.

Concerning the correctness, the important point is that, if~$G'$ is a circle graph and the marker corresponding to the parent of a prime node~$u$ is in $MP(u)$, then by \autoref{lem_fully_mixed}, one of its consecutive chords (in $\CSC(u)$) is also in~$MP(u)$, and thus the parent chord will be identified at some point during the traversal.  Thus, if~$G'$ is a circle graph the algorithm correctly computes the tree-edges of~$T(N_{G'}(x))$.
\end{proof}

We now have to mark the tree-edge extremities of $\PC(G)$, each of which is either a marker vertex of a degenerate node or an endpoint of a chord in a CSC. Given that $T(N_{G'}(x))$ has been identified, the marking procedure that is fully described in~\cite{GioanPTC14Split} (Lemma 5.10 therein) runs in $O(|T(N_{G'}(x))|)$ time. Within the same time complexity, we can identify which of the cases described in \autoref{sub_split_tree_update} holds. At that step, the node-split operation on degenerate nodes is required. Let us remark that, as in a split PC-tree every degenerate node maintains a pointer to its parent node, we can perform the node-split as described in~\cite{GioanPTC14Split}.

\begin{lemma}[\cite{GioanPTC14Split}] \label{lem_degenerate_split}
Let $u$ be a degenerate node of the split PC-tree $\PC(G)$ of a circle graph. If $(A,B)$ is a split of $G(u)$, then node-splitting $u$ according to $(A,B)$ can be done in $O(|A|)$-time.
\end{lemma}

Let us now describe how, having computed $T(N_{G'}(x))$ and marked $\PC(G)$ according to $N_{G'}(x)$, we can efficiently compute $\PC(G+x)$.

\begin{lemma} \label{lem_update}
Let $G'=G+x$ be a graph such that $x$ is a good vertex of $G'$ and $G$ is a circle graph. 
Assume that $\PC(G)$ and $T(N_{G'}(x))$ are given and that every edge extremity is marked according to $N_{G'}(x)$. In time $O(|T(N_{G'}(x))|)$, we can either compute $\PC(G')$ or conclude that $G'$ is not a circle graph.
\end{lemma}
\begin{proof}
Let us consider the four distinct cases identified in \autoref{sub_split_tree_update}.
\begin{enumerate}
\smallskip
\item \emph{$T$ contains a clique node $u$ whose marker vertices are all perfect; or $T$ contains a star node $u$ whose marker vertices are all empty except its center, which is perfect; or $T$ contains a unique hybrid node~$u$, which is prime.}

\smallskip
Suppose that $u$ is prime, since otherwise $\PC(G')$ is obtained by adding a leaf to a degenerate node which can be done in $O(1)$-time. We need to test if the chord endpoints of $MP(u)$ are consecutive in the CSC $\CD(u)$. This can be done in $O(|MP(u)|)$-time by \autoref{lem_consecutive_test}. 
If the test fails, we  can conclude that $G+x$ is not a circle graph. Otherwise we insert the new chord corresponding to $x$, which by \autoref{lem_chord_insertion} can be done in $O(1)$-time. Observe that if a marker vertex (or a chord endpoint) is perfect or mixed, then its incident tree-edge belongs to $T(N_{G'}(x))$. So we have $|MP(u)|\leqslant |T(N_{G'}(x))|$ implying that the global complexity is $O(|T(N_{G'}(x))|)$.

\smallskip
\item \emph{$T$ contains a tree-edge~$e$ whose extremities are both perfect, and this edge is unique; or $T$ contains a tree-edge with one perfect and one empty extremity, and this edge is unique.}

\smallskip
In this case, $\PC(G')$ is obtained by subdividing $e$ by a new ternary degenerate node. This can clearly be achieved in $O(1)$-time.

\smallskip
\item \emph{$T$ contains a hybrid node $u$, and this node is degenerate.}

\smallskip
In this case, $\PC(G')$ is obtained by first performing a node-split of $u$ according to $(P^*(u),E^*(u))$ and then inserting a new degenerate node. By \autoref{lem_degenerate_split} and the discussion above, this can be done in $O(|P(u)|)$-time. Since every perfect marker vertex of $u$ is incident to a tree-edge of $T(N_{G'}(x))$, the statement holds.

\smallskip
\item \emph{$T$ contains a (unique) fully mixed subtree $M$.}

\smallskip
We first have to perform the cleaning step. By \autoref{lem_degenerate_split}, this can be done in global $O(|T(N_{G'}(x))|)$ time.  Then we construct a CSC for each mixed degenerate node, which takes global $O(|T(N_{G'}(x))|)$ time, since each of them has constant size.  Finally, by \autoref{lem_consecutive_join}, the contraction of the fully mixed subtree of $\mathsf{cl}(\PC(G))$ can also be performed in $O(|T(N_{G'}(x))|)$-time. Observe that the contraction step may fail if some node join cannot return a chord diagram satisfying the consecutive property. Finally, by \autoref{lem_chord_insertion}, inserting a new chord to the CSC resulting from the series of contraction can be achieved in $O(1)$-time. \qedhere
\end{enumerate}
\end{proof}

\subsection{Amortized time complexity analysis}

A key ingredient in the complexity analysis is the following result, which bounds the size of $T(N_G(x))$ in terms of the size of~$N_G(x)$.
\begin{lemma}[\cite{RoseTL76Algorithmic,Spinrad94Recognition}]
  \label{lem_neighborhood_size}
Let $\ST(G) = (T,\mathcal F)$ be the split-tree of a graph $G$. For every vertex $x$ of $G$, $|T(N_G(x))| \le 2 \cdot |N_G(x)|$.
\end{lemma}

Given the split PC-tree $\PC(G)$ of a circle graph $G$, testing if $G'=G+x$ is a circle graph amounts to computing $\PC(G')$. For this we have access to $T(N_{G'}(x))$, with $\ST(G)=(T,\F)$, and not to $T'(N_{G'}(x))$, with $\ST(G')=(T',\F')$. Observe that  \autoref{lem_update} establishes an insertion complexity in time linear in the size of $T(N_{G'}(x))$. This is not sufficient to establish an overall linear running time. Indeed, the size of $T(N_{G'}(x))$ can be significantly larger than $|N_{G'}(x)|$. 
In that case, \autoref{lem_neighborhood_size} says that $|T'(N_{G'}(x))|$ is significantly smaller than~$|T(N_{G'}(x))|$. This is the case when a large fully-mixed subtree is contracted into a single prime node. So to circumvent this issue, we show that inserting a good vertex $x$ in a circle graph $G$ takes amortized time~$O(|N_{G'}(x)|)$ if~$G'=G+x$ is a circle graph. We use the method of potentials introduced in~\cite{SleatorT85}; see~\cite[Ch. 17.3]{cormen01introduction} for a modern exposition.

We define a potential function~$\Phi$ on \GLT{}'s representing a split decomposition.  Let~$(T,\mathcal F)$ be a \GLT.  For an inner node~$u$ of $T$, we define :

\smallskip
\centerline{$
\Phi_u= \left \{
\begin{array}{ll}
1 &\mbox{ if $u$ is non-degenerate,} \\
\deg_{T}(u)-2 &\mbox{ if $u$ is degenerate.}
\end{array}
\right .
$}

\noindent
\smallskip
And we set $\Phi(T,\mathcal F) = \sum_{u \in T} \Phi_u$. Note that, by definition, the potential is non-negative. 
The \emph{amortized cost} to insert vertex $x$ is

\smallskip
\centerline{$a(x)=t(x)+\Phi(\ST(G+x))- \Phi(\ST(G)),$}

\smallskip
\noindent
where $t(x)$ is the \emph{actual cost}, which is given by \autoref{lem_update}, that is $O(|T(N_{G'}(x))|)$.

The intuition behind the definition of this potential function is the following. When a large fully-mixed subtree into contracted into a single node, which is a costly operation, then the potential will significantly decrease. By contrast, when a new degenerate node or new leaf adjacent to a degenerate node is added, which can be done efficiently, then the potential increases, but only slightly. This implies the overall amortized linear running time.

\begin{lemma} \label{lem_amortized}
Let $G'=G+x$ be a graph such that $x$ is a good vertex of $G$ and $G$ is a circle graph.
If $G'$ is a circle graph, then~$\PC(G')$ can be computed from~$\PC(G)$ in amortized~$O(|N_{G'}(x)|)$ time.
\end{lemma}

\begin{proof}
Let us denote $\ST(G)=(T,\F)$ and $\ST(G')=(T',\F')$.  By \autoref{lem_update}, the actual cost $t(x)$ is 
  in~$O(|T(N_{G'}(x))|)$.  We now consider the four cases identified in
  \autoref{sub_split_tree_update} and in each of them determine the
  potential difference and the resulting amortized running time.
\begin{enumerate}
  \item As~$\ST(G')$ is obtained from~$\ST(G)$ by either adding a
    chord to a prime node or by adding a new leaf, we
    have~$\Phi(\ST(G')) - \Phi(\ST(G)) \leq 1$.  Moreover, we
    have~$|T(N_{G'}(x))| = |T'(N_{G'}(x))|$. So by \autoref{lem_neighborhood_size}, $|T(N_{G'}(x))| \le 2 \cdot |N_{G'}(x)|$. 
    Thus, $t(x)$ is in $O(|N_{G'}(x)|)$ and $\Phi(\ST(G+x))- \Phi(\ST(G))$ is constant, which yields an amortized
    running time $a(x)$ in~$O(|N_{G'}(x)|)$.

  \smallskip
  \item In this case~$\ST(G')$ is obtained from~$\ST(G)$ by subdividing
    an edge with a new degenerate node of degree $3$.
    Thus we have~$\Phi(\ST(G')) - \Phi(\ST(G)) =1$ and
    $|T(N_{G'}(x))| = |T'(N_{G'}(x))|-1$. So by \autoref{lem_neighborhood_size}, $|T(N_{G'}(x))| = |T'(N_{G'}(x))|-1 \le 2 \cdot |N_{G'}(x)| -1$.
    Thus, $t(x)$ is in $O(|N_{G'}(x)|)$ and $\Phi(\ST(G+x))- \Phi(\ST(G))$ is constant, which yields an amortized
    running time $a(x)$ in~$O(|N_{G'}(x)|)$.

  \smallskip
  \item In this case~$\ST(G')$ is obtained from~$\ST(G)$ by splitting a hybrid
    node~$u$ that is degenerate into two adjacent nodes~$v,w$ and then using the previous
    case to handle the edge~$vw$. Observe that performing a node-split on a degenerate node does not change the potential. Moreover as discussed in the previous case, subdividing a tree-edge to insert a degenerate node of degree $3$ increases the potential of the resulting GLT by $1$. It follows that 
    $\Phi(\ST(G')) - \Phi(\ST(G)) = 1$. Observe that $|T(N_{G'}(x))| = |T'(N_{G'}(x))| - 2$, then by \autoref{lem_neighborhood_size}, $|T(N_{G'}(x))| = |T'(N_{G'}(x))| - 2 \le 2 \cdot |N_{G'}(x)| -2$.
Thus, $t(x)$ is in $O(|N_{G'}(x)|)$ and $\Phi(\ST(G+x))- \Phi(\ST(G))$ is constant, which yields an amortized
    running time $a(x)$ in~$O(|N_{G'}(x)|)$.
   
\smallskip    
\item Suppose that the fully-mixed subtree of $\ST(G)$ contains $d$ degenerate nodes and $p$ prime nodes. Since each node-split of a degenerate node leaves the potential unchanged, we have $\Phi(\cl(\ST(G))= \Phi(\ST(G))$. Observe that in $\cl(\ST(G))$, every degenerate node $u$ has degree~$3$, implying that $\Phi_u=1$. Thus every node (degenerate or prime) of the fully mixed subtree of $\cl(\ST(G))$ contributes $1$ to $\Phi(\cl(\ST(G)))$. Since $\ST(G')$ is obtained by contracting the fully-mixed subtree of $\cl(\ST(G))$ to a single prime node, it follows that $\Phi(\ST(G'))=\Phi(\cl(\ST(G)))-d-p+1$. So we have $\Phi(\ST(G'))-\Phi(\ST(G))=-d-p+1$.
 Let us now analyze the size of $T'(N_{G'}(x))$ compared to $T(N_{G'}(x))$.
Observe that a degenerate node may be split twice, but then in the contraction step, one of the up to three resulting nodes disappears in $T'(N_{G'}(x))$. Moreover every prime node of $T(N_{G'}(x))$ also disappears during the contraction step. It follows that $|T(N_{G'}(x))|\leq |T'(N_{G'}(x))|+d+p-1$.
Then, by \autoref{lem_neighborhood_size}, $|T(N_{G'}(x))| \leq 2 \cdot |N_{G'}(x)| +d+p-1$.  It follows that the amortized running time is $a(x)=|T(N_{G'}(x))|+\Phi(\ST(G'))-\Phi(\ST(G))\leq 2 \cdot |N_{G'}(x)|$. \qedhere
  \end{enumerate}
\end{proof}

\begin{theorem}
Let $G$ be a graph on $n$ vertices and $m$ edges. Deciding if $G$ is a circle graph can be done in $O(n+m)$.
\end{theorem}
\begin{proof}
Since the correctness of the algorithm is proved in Gioan \etal~\cite{GioanPTC14Circle}, we only discuss the complexity analysis. 
If $G$ is a  circle graph, then by \autoref{lem_amortized}, the algorithm builds the split PC-tree $\PC(G)$ in amortized linear time. So suppose that $G$ is not a circle graph. Let $\sigma$ be the LexBFS ordering in which the vertices are processed by the algorithm and $x$ be the earliest vertex in $\sigma$ and $S$ be the subset of vertices containing all vertices $y$ such that $y<_{\sigma} x$. Then $G[S]$ is a circle graph and we note $\ST(G[S])=(T_S,\F_S)$. The fact that $G[S]+x$ is not a circle graph is detected either during the computation of $T_S(N(x))$ or when we try to insert the chord corresponding to $x$ to a chord diagram representing a prime node in the spit PC-tree in hand. By \autoref{lem_compute_subtree} and by \autoref{lem_update}, this can be decided in time $O(|T_S(N(x))|)$ which is compatible with the amortized complexity analysis of \autoref{lem_amortized}.
\end{proof}


\bibliography{references}

@article{Bouchet72Caracterisation,
	author = {A. Bouchet},
	journal = {Compte-rendus Acad{\'e}mie des Sciences ({CRAS})},
	pages = {724-727},
	title = {Caract{\'e}risation des symboles crois{\'e}s de genre nul},
	volume = {274},
	year = {1972}}

@article{Bouchet87Reducing,
	author = {A. Bouchet},
	doi = {10.1007/BF02579301},
	journal = {Combinatorica},
	pages = {243-254},
	title = {Reducing prime graphs and recognizing circle graphs},
	volume = {7},
	year = {1987}}

@article{Bouchet94Circle,
	author = {A. Bouchet},
	doi = {10.1006/jctb.1994.1008},
	journal = {Journal of Combinatorial Theory Series B},
	pages = {107-144},
	title = {Circle graph obstructions},
	volume = {60},
	year = {1994}}

@article{BrucknerRS24Extending,
	author = {G. Br{\"u}ckner and I. Rutter and P. Stumpf},
	doi = {10.1007/s00453-024-01216-5},
	journal = {Algorithmica},
	pages = {2152-2173},
	title = {Extending partial representation of circle graphs in near-linear time},
	volume = {86},
	number       = {7},
	year = {2024}}

@book{cormen01introduction,
  author = {T.H. Cormen and C.E. Leiserson and R.L. Rivest and C. Stein},
  edition = {3rd},
  publisher = {The MIT Press},
  title = {Introduction to Algorithms},
  year = 2001
}

@article{CorneilOS09TheLBFS,
	author = {D.G. Corneil and S. Olariu and L. Stewart},
	doi = {10.1137/S089548010037345},
	journal = {SIAM Journal on Computing},
	number = {4},
	pages = {1905-1953},
	title = {The {LBFS} structure and recognition of interval graphs},
	volume = {23},
	year = {2009}
}

@article{Courcelle08Circle,
	author = {B. Courcelle},
	doi = {10.1016/j.jal.2007.05.001},
	journal = {Journal of Applied Logic},
	pages = {416-442},
	title = {Circle graphs and monadic second-order logic},
	volume = {6},
	year = {2008}
}

@article{CunnighamE80Acombinatorial,
	author = {W.H. Cunningham and J. Edmonds},
	doi = {10.4153/CJM-1980-057-7},
	journal = {Canadian Journal of Mathematics},
	number = {3},
	pages = {734-765},
	title = {A combinatorial decomposition theory},
	volume = {32},
	year = {1980}
}

@inproceedings{Dahlhaus94SplitConf,
	author = {E. Dahlhaus},
	booktitle = {Foundations of Software Technology and Theoretical Computer Science - FSTTCS},
	doi = {10.1007/3-540-58715-2_123},
	pages = {171-180},
	series = {Lecture Notes in Computer Science},
	title = {Efficient Parallel and Linear Time Sequential Split Decomposition (Extended Abstract).},
	volume = {880},
	year = {1994}}

@article{Dahlaus00SplitJournal,
	author = {E. Dahlhaus},
	doi = {10.1006/jagm.2000.1090},
	journal = {Journal of Algorithms},
	number = {2},
	pages = {205-240},
	title = {Parallel Algorithms for Hierarchical Clustering and Applications to Split Decomposition and Parity Graph Recognition.},
	volume = {36},
	year = {2000}
}

@article{deFraysseix84Acharacterization,
	author = {H. de Fraysseix},
	doi = {10.1016/S0195-6698(84)80005-0},
	journal = {European Journal of Combinatorics},
	pages = {223-238},
	title = {A characterization of circle graphs},
	volume = {5},
	year = {1984}
}

@article{EvenI71Queues,
	author = {S. Even and A. Itai},
	doi = {10.1016/B978-0-12-417750-5.50011-7},
	journal = {Theory of Machines and Computations},
	pages = {71--86},
	title = {Queues, stacks and graphs},
	year = {1971}
}

@article{FinkPR23Experimental,
	author = {S.D. Fink and M. Pfrestzschner and I. Rutter},
	doi = {10.1145/3611653},
	journal = {ACM Journal of Experimental Algorithms},
	pages = {24},
	title = {Experimental comparison of {PC}-trees and {PQ}-trees},
	volume = {28},
	year = {2023}
}

@article{Fournier78UneCaracterisation,
	author = {J.C. Fournier},
	journal = {Compte-rendus Acad{\'e}mie des Sciences ({CRAS})},
	pages = {811-813},
	title = {Une caract{\'e}risation des graphes de cordes},
	volume = {286A},
	year = {1978}
}

@article{GaborHS89Recognizing,
	author = {C.P. Gabor and W.L. Hsu and K.J. Suppovit},
	doi = {10.1145/65950.65951},
	journal = {Journal of ACM},
	pages = {435-473},
	title = {Recognizing circle graphs in polynomial time},
	volume = {36},
	year = {1989}
}

@article{GallerF64AnImproved,
	author = {B.A. Galler and M.J. Fischer},
	doi = {10.1145/364099.364331},
	journal = {Communications of the ACM},
	number = {5},
	pages = {301-303},
	title = {An improved equivalence algorithm},
	volume = {7},
	year = {1964}
}

@article{GeelenKMcCW23TheGrid,
	author = {J. Geelen and O-j. Kwon and R. McCarty and P. Wollan},
	doi = {https://doi.org/10.1016/j.jctb.2020.08.004},
	journal = {Journal of Combinatorial Theory, Series B},
	pages = {93-116},
	title = {The Grid Theorem for vertex-minors},
	volume = {158},
	year = {2023}
}

@article{GioanP12Split,
	author = {E. Gioan and C. Paul},
	doi = {10.1016/j.dam.2011.05.007},
	journal = {Discrete Applied Mathematics},
	pages = {708-723},
	title = {Split decomposition and graph-labelled trees: characterizations and fully dynamic algorithms for totally decomposable graphs},
	volume = {160},
	year = {2012}
}

@article{GioanPTC14Split,
	author = {E. Gioan and C. Paul and M. Tedder and D.G Corneil},
	doi = {10.1007/s00453-013-9752-9},
	journal = {Algorithmica},
	pages = {789--843},
	title = {Practical and Efficient Split Decomposition via Graph-Labelled Trees},
	volume = 69,
	year = 2014
}

@article{GioanPTC14Circle,
	author = {E. Gioan and C. Paul and M. Tedder and D.G Corneil},
	doi = {10.1007/s00453-013-9745-8},
	journal = {Algorithmica},
	pages = {759--788},
	title = {Practical and Efficient Circle Graph Recognition},
	volume = 69,
	year = 2014
}

@article{HsuMcC03PCTrees,
	author = {W.-L. Hsu and R. McConnell},
	doi = {10.1016/S0304-3975(02)00435-8},
	journal = {Theoretical Computer Science},
	pages = {99-116},
	title = {{PC} trees and cicular-ones arrangements},
	volume = {296},
	year = {2003}
}

@Article{Hsu95,
  author = 	 {W.L. Hsu},
  title = 	 {${O}(m \cdot n)$ algorithms for the recognition and isomorphism problems on circular-arc},
  journal = 	 {SIAM Journal on Computing},
  doi= {10.1137/S0097539793260726},
  year = 	 1995,
  volume = 	 24,
  number = 	 3,
  pages = 	 {411--439}
 }

@inproceedings{KaliszKZ22,
  author       = {V. Kalisz and P. Klav{\'{\i}}k and P. Zeman},
  title        = {Circle Graph Isomorphism in Almost Linear Time},
  booktitle    = {Annual Conference on Theory and Applications of Models of Computation - {TAMC}},
  series       = {Lecture Notes in Computer Science},
  volume       = {13571},
  pages        = {176--188},
  year         = {2022},
  doi          = {10.1007/978-3-031-20350-3\_15},
}

@inproceedings{Kotzig77Quelques,
	author = {A. Kotzig},
	booktitle = {S{\'e}minaire de Paris},
	title = {Quelques remarques sur les transformations {K}},
	year = {1977}
}

@article{MaS94QuadraticSplit,
	author = {T.-H. Ma and J. Spinrad},
	doi = {10.1006/jagm.1994.1007},
	journal = {Journal of Algorithms},
	pages = {145-160},
	title = {An ${O}(n^2)$ algorithm for undirected split decomposition},
	volume = {16},
	year = {1994}
}

@article{Naji85Reconnaissance,
	author = {W. Naji},
	doi = {10.1016/0012-365X(85)90117-7},
	journal = {Discrete Mathematics},
	pages = {329-337},
	title = {Reconnaissance des graphes de cordes},
	volume = {54},
	year = {1985}
}

@article{Oum05Rankwidth,
	author = {S.I. Oum},
	doi = {10.1016/j.jctb.2005.03.003},
	journal = {Journal of Combinatorial Theory Series B},
	pages = {79-100},
	title = {Rank-width and vertex-minors},
	volume = {95},
	year = {2005}
}

@article{RoberstonS90Treewidth,
	author = {N. Robertson and P. Seymour},
	doi = {10.1016/0095-8956(90)90120-O},
	journal = {Journal of Combinatorial Theory Series B},
	pages = {227-254},
	title = {Graph minors {IV}. Tree-width and well-quasi-ordering},
	volume = {48},
	year = {1990}
}

@article{RoberstonS86Excluding,
	author = {N. Robertson and P. Seymour},
	doi = {10.1016/0095-8956(86)90030-4},
	journal = {Journal of Combinatorial Theory Series B},
	pages = {92-114},
	title = {Graph minors {V}: excluding a planar graph},
	volume = {41},
	year = {1986}
}

@article{RoseTL76Algorithmic,
	author = {D.J. Rose and R.E. Tarjan and G.S. Lueker},
	doi = {10.1137/0205021},
	journal = {SIAM Journal on Computing},
	number = {2},
	pages = {266-283},
	title = {Algorithmic aspects of vertex elimination on graphs},
	volume = {5},
	year = {1976}
}

@article{ShihH99Anew,
	author = {W. K. Shih and W. L. Hsu},
	doi = {10.1016/S0304-3975(98)00120-0},
	journal = {Theoretical Computer Science},
	pages = {179--191},
	title = {A new planarity test},
	volume = {223},
	year = {1999}
}

@article{SleatorT85,
  author       = {D.D. Sleator and R.E. Tarjan},
  title        = {Amortized Efficiency of List Update and Paging Rules},
  journal      = {Communication of the {ACM}},
  volume       = {28},
  number       = {2},
  pages        = {202--208},
  year         = {1985},
  doi          = {10.1145/2786.2793},
}

@article{Spinrad94Recognition,
	author = {J. Spinrad},
	doi = {10.1006/jagm.1994.1012},
	journal = {Journal of Algorithms},
	pages = {264-282},
	title = {Recognition of circle graphs},
	volume = {16},
	year = {1994}
}

@article{Tarjan75Efficiency,
	author = {R. Tarjan},
	doi = {10.1145/321879.321884},
	journal = {Journal of the ACM},
	pages = {215-225},
	title = {Efficiency of a good but not linear set union algorithm},
	volume = {22},
	year = {1975}
}

\end{document}